\documentclass[11pt]{article}
\usepackage[margin=1in]{geometry}
\usepackage{pig}
\usepackage{graphicx}
\usepackage{enumerate}

\usepackage{authblk}
\title{How Robust are Reconstruction Thresholds for Community Detection?}

\author[1,2]{Ankur Moitra\thanks{Email: {\tt moitra@mit.edu}. This work is supported in part by NSF CAREER Award CCF-1453261, a grant from the MIT NEC Corporation and a Google Faculty Research Award.}}
\author[1]{William Perry\thanks{Email: {\tt wperry@mit.edu}.}}
\author[1]{Alexander S.\ Wein\thanks{Email: {\tt awein@mit.edu}. This work is supported in part by an NDSEG graduate fellowship.}}
\affil[1]{Massachusetts Institute of Technology, Department of Mathematics}
\affil[2]{Massachusetts Institute of Technology, Computer Science and Artificial Intelligence Lab}

\begin{document}

\begin{titlepage}
\maketitle\thispagestyle{empty} 

\begin{abstract}
The stochastic block model is one of the oldest and most ubiquitous models for studying clustering and community detection. In an exciting sequence of developments, motivated by deep but non-rigorous ideas from statistical physics, Decelle et al.\ \cite{decelle} conjectured a sharp threshold for when community detection is possible in the sparse regime. Mossel, Neeman and Sly \cite{mns} and Massouli\'e \cite{massoulie} proved the conjecture and gave matching algorithms and lower bounds. 

Here we revisit the stochastic block model from the perspective of semirandom models where we allow an adversary to make `helpful' changes that strengthen ties within each community and break ties between them. We show a surprising result that these `helpful' changes can shift the information-theoretic threshold, making the community detection problem strictly harder. We complement this by showing that an algorithm based on semidefinite programming (which was known to get close to the threshold) continues to work in the semirandom model (even for partial recovery). This suggests that algorithms based on semidefinite programming are robust in ways that {\em any} algorithm meeting the information-theoretic threshold 
cannot be.

These results point to an interesting new direction: Can we find robust, semirandom analogues to some of the classical, average-case thresholds in statistics? We also explore this question in the broadcast tree model, and we show that the viewpoint of semirandom models can help explain why some algorithms are preferred to others in practice, in spite of the gaps in their statistical performance on random models.
\end{abstract}

\end{titlepage}

\section{Introduction}

\subsection{Background}

The stochastic block model is one of the oldest and most ubiquitous models for studying clustering and community detection. It was first introduced by Holland et al.\ \cite{holland} more than thirty years ago and since then it has received considerable attention within statistics, computer science, statistical physics and information theory. 
The model defines a procedure to generate a random graph: first, each of the $n$ nodes is independently assigned to one of $r$ communities, where $p_i$ is the probability of being assigned to community $i$. Next, edges are sampled independently based on the community assignments: if nodes $u$ and $v$ belong to communities $i$ and $j$ respectively, the edge $(u,v)$ occurs with probability $Q_{ij}$ independent of all other edges, where $Q$ is an $r \times r$ symmetric matrix. The goal is to recover the community structure --- either exactly or approximately --- from the graph.

Since its introduction, the stochastic block model has served as a testbed for the diverse range of algorithms that have been developed for clustering and community detection, including combinatorial methods \cite{bui}, spectral methods \cite{mcsherry}, MCMC \cite{jerrum}, semidefinite programs \cite{boppana} and belief propagation \cite{MNSbelief}. Recently, the stochastic block model has been thrust back into the spotlight, the goal being to establish tight thresholds for when community detection is possible, and to find algorithms that achieve them. Towards this end, Decelle et al.\ \cite{decelle} made some fascinating conjectures (which have since been resolved) in the case of two equal-sized communities with constant average degree, that we describe next.

Throughout this paper, we will focus on the two-community case and use $a/n$ and $b/n$ to denote the within-community and between-community connection probabilities respectively. We will be interested in the sparse setting where $a, b = \Theta(1)$. Moreover we set $p_1 = p_2 = 1/2$, so that the two communities are roughly equal-sized, and we assume $a > b$ (although we discuss the case $a < b$ in Appendix~\ref{sec:dissortative}). We will use $G(n, a/n, b/n)$ to denote the corresponding random graph model. The setting of parameters above is particularly well-motivated in practice, where a wide variety of networks have been observed \cite{leskovec} to have average degree that is bounded by a small constant. It is important to point out that when the average degree is constant, it is impossible to recover the community structure exactly in the stochastic block model because
a constant fraction of nodes will be isolated. Instead, our goal is to recover a partition of the nodes into two communities that has non-trivial agreement (better than random guessing) with the true communities as $n \rightarrow \infty$, and we refer to this as {\em partial recovery}.

\begin{conjecture} \cite{decelle}
If $(a - b)^2 > 2(a + b)$ then partial recovery in $G(n, a/n, b/n)$ is possible, and if $(a - b)^2 < 2(a + b)$ then partial recovery is information-theoretically impossible. 
\end{conjecture}

\noindent This conjecture was based on deep but non-rigorous ideas originating from statistical physics, and was first derived heuristically as a stability criterion for belief propagation. This threshold also bears a close connection to known thresholds in the broadcast tree model, which is another setting in which to study partial recovery. We formally define the broadcast tree model in Section~\ref{sec:prelim}, but it is a stochastic process described by two parameters $a$ and $b$. Kesten and Stigum \cite{kesten-stigum} showed that partial recovery is possible if $(a - b)^2 > 2(a + b)$, and much later Evans et al.\ \cite{evans} showed that it is impossible if $(a - b)^2 \leq 2(a + b)$. The connection between the two models is that in the stochastic block model, the local neighborhood of a node resembles the broadcast tree model, and this was another compelling motivation for the conjecture.

In an exciting sequence of developments, Mossel, Neeman and Sly \cite{mns} proved a lower bound that even distinguishing $G(n, a/n, b/n)$ from an Erd\H{o}s--R\'enyi graph $G(n, (a+b)/2n)$ is information-theoretically impossible if $(a - b)^2 \leq 2(a + b)$, by a careful coupling to the broadcast tree model. Subsequently Mossel, Neeman and Sly \cite{mns2} and Massouli\'e \cite{massoulie} independently gave matching algorithms that achieve partial recovery up to this threshold, thus resolving the conjecture! Mossel, Neeman and Sly \cite{MNSbelief} later showed that for some constant $C > 0$, if $(a-b)^2 > C (a + b)$ then belief propagation works and moreover the agreement of the clustering it finds is the best possible.

In fact, many other sorts of threshold phenomena have been found in different parameter regimes. Abbe, Bandeira and Hall \cite{abh} studied exact recovery in the logarithmic degree setting $G(n, a\log n/n, b \log n/n)$ and showed that it is efficiently possible to recover the two communities exactly if $(\sqrt{a} - \sqrt{b})^2 > 2$ and information-theoretically impossible if $(\sqrt{a} - \sqrt{b})^2 < 2$. Abbe and Sandon \cite{as} gave a precise characterization of when exact recovery is possible in the general stochastic block model for more than two communities with arbitrary relative sizes and arbitrary connection probabilities.

\subsection{Semirandom Models}

This abundance of sharp thresholds begs a natural question: how robust are these reconstruction thresholds? It is clear that if one substantially changes the distributional model, the thresholds themselves are likely to change. However there is a subtler issue. The algorithms that achieve these thresholds may in principle be over-fitting to a particular distributional model. Random graphs are well-known \cite{alon-spencer} to have rigid properties, such as sharp laws for the distribution of subgraph counts and a predictable distribution of eigenvalues. Real-world graphs do not have such properties.

In a remarkable paper, Blum and Spencer \cite{bs} introduced the semirandom model as an intermediary between average-case and worst-case analysis, to address such issues. The details of the model vary depending on the particular optimization problem, and since we will focus on clustering we will be most interested in the variant used by Feige and Kilian \cite{fk}.

\vspace{1em}\fbox{\parbox{0.9\textwidth}{
{\bf Semirandom Model for Community Detection:} \cite{fk}
\begin{itemize}
\item Sample a `precursor' graph $\Gpre$ from $G(n, a/n, b/n)$.
\item A monotone `adversary' observes $\Gpre$ along with the hidden community structure, and can delete any number of edges crossing between the two communities, and add any number of edges with each community.
\item Output the resulting graph $G$.
\end{itemize}
}}
\vspace{1em}

\noindent The adversary above is called `monotone' because it is restricted to making changes that seem to be helpful. It can only strengthen ties within each community, and break ties between them.\footnote{This is the reason we require $a > b$. If we had $b > a$ we could define the adversary in the opposite way (which we analyze in Appendix~\ref{sec:dissortative}).} The key is that a monotone adversary can break the sorts of rigid structures that arise in random graphs, such as predictable degree distributions and subgraph counts. An algorithm that works in a semirandom model can no longer rely on these properties. Instead of a graph containing {\em only} random edges, all we are assured of is that it contains {\em some} random edges.

In this paper, we use semirandom models as our notion of `robustness.' Many forms of robustness exist in the literature -- for example, the independent work \cite{mmv-robust} gives algorithms that are robust to $o(n)$ non-monotone changes in addition to monotone changes. With most robustness models, any algorithm will break down after enough errors, and one can compare algorithms based on how many errors they can tolerate. In contrast, semirandom models distinguish between algorithms {\em qualitatively}: as we will see, entire classes of algorithms continue to work under any number of monotone changes, while others do not.

Feige and Kilian \cite{fk} showed that semidefinite programs for \emph{exact recovery} in the stochastic block model continue to work in the semirandom model --- in fact they succeed up to the threshold for the random model \cite{hwx}. Since then, there have been many further developments including algorithms that work in semirandom models for planted clique \cite{fkr}, unique games \cite{kmm}, various graph partitioning problems \cite{mmv}, correlation clustering \cite{msc, mmv15} and the general planted partition model\footnote{In the general planted partition model, we have $Q_{ii} = a/n$ for all $i$ and $Q_{ij} = b/n$ for all $i \neq j$, with $a > b$, but the number of communities and their relative sizes are arbitrary constants.} \cite{pw} (and in some cases for even more powerful adversaries). A common theme is that if you have a semidefinite program that works in some stochastic setting, then it often extends almost automatically to the semirandom setting. So are there semidefinite programs that achieve the same sharp partial recovery results as Mossel, Neeman and Sly \cite{mns2} and Massouli\'e \cite{massoulie}, and that extend to the semirandom model too? Or is there a genuine gap between what is achievable in the random vs.\ the semirandom setting?

\subsection{Our Results}
Recall that in the semirandom block model, a monotone adversary observes a sample from the stochastic block model and is then allowed to add edges within each community and delete edges that cross between communities. We will design a particularly simple adversary to prevent an algorithm from utilizing the paths of length two that go from a `$+$' node to a `$-$' node and back to a `$+$' node, where the middle node has degree two. Our adversary will delete any such path it finds (with some additional technical conditions to locally coordinate these modifications), and our main result is that, surprisingly, this simple modification strictly changes the partial recovery threshold. We will state our bounds in terms of the average degree $k = \frac{a + b}{2}$ and `noise' $\eps = \frac{b}{a + b} \in [0,\frac12)$, in which case the threshold $(a-b)^2 > 2(a+b)$ becomes $k(1-2 \eps)^2 > 1$. Note that this threshold requires $k > 1$. Then:

\begin{theorem}[main]\label{thm:intro-main}
For any $k > 1$, there exists $0 < \eps < \frac12$ so that $k (1- 2\eps)^2 > 1$ and hence partial recovery in the stochastic block model $G(n, a/n, b/n)$ is possible, and yet there is a monotone adversary so that partial recovery in the semirandom model is information-theoretically impossible.
\end{theorem}

\noindent A common tool in the algorithms of Mossel, Neeman and Sly \cite{mns2} and of Massouli\'e \cite{massoulie} is the use of non-backtracking walks and spectral bounds for their transition matrices. Our adversary explicitly deletes a significant number of these walks. This simple modification not only defeats these particular algorithms, but we can show that {\em no} algorithm can achieve partial recovery up to the threshold $k (1-2\eps)^2 > 1$. To the best of our knowledge, this is the first explicit separation between what is possible in the random model vs.\ the semirandom model, for {\em any} problem with a monotone adversary. 

We show a complementary result, that no monotone adversary can make the problem too much harder. Various semidefinite programs have been designed for partial recovery.
These algorithms work in the constant-degree regime, but are not known to work all the way down to the information-theoretic threshold. In particular, we follow Gu\'edon and Vershynin \cite{GV} and show that their analysis (with a simple modification) works as is for the semirandom model. This shows that semidefinite programs not only give rise to algorithms that work for the semirandom model in the exact recovery setting, but also for partial recovery too under some fairly general conditions. 

\begin{theorem}\label{thm:intro-robust-sdp}
Let $k = \frac{a + b}{2}$ and $\eps = \frac{b}{a + b} < \frac12$. There is a constant $C > 1$ so that if $a > 20$ and $k (1- 2\eps)^2 > C$ then partial recovery is possible in the semirandom block model. Moreover it can be solved in polynomial time. 
\end{theorem}

\noindent Our robustness proof only applies to a particular form of SDP analysis. Given a different proof that the SDP succeeds in the random model for a larger range of parameters $k,\eps$ than above, there would be no guarantee that the SDP also succeeds in the semirandom model for that range of parameters. Hence we cannot formally conclude that it is impossible for the SDP to reach the information-theoretic threshold in the random model, though our results are suggestive in this direction. 

We remark that each possible monotone adversary yields a new distribution on planted community detection problems. Hence, we can think of the algorithm in the theorem above as one that performs almost as well as information-theoretically possible across an entire family of distributions, simultaneously. This is a major advantage to algorithms based on semidefinite programming, and points to an interesting new direction in statistics. Can we move beyond average-case analysis? Can we find robust, semirandom analogues to some of the classical, average-case thresholds in statistics? The above two theorems establish upper and lower bounds for this semirandom threshold, that show that it is genuinely different from the average-case threshold. The usual notion of a threshold makes sense when we have exact knowledge of the process generating the instances we wish to solve. But when we lack this knowledge, semirandom models offer an avenue for exploration that can lead to new algorithmic and statistical questions.

Along the way to proving our main theorem, we show a random vs.\ semirandom separation for the broadcast tree model too. We define this model in Section~\ref{sec:prelim}. In short, it is a stochastic process in which each node is given one of two labels and gives birth to $\Pois(a/2)$ nodes of the same label and $\Pois(b/2)$ nodes of the different label, with the goal being to guess the label of the root given the labels of the leaves. There are many ways we could define a monotone adversary, and we focus on a particularly weak one that is only allowed to cut edges between nodes with different labels. We call this the cutting adversary, and we prove:

\begin{theorem}\label{thm:intro-treehard}
For any $k > 1$, there exists $0 < \eps < \frac12$ so that $k (1- 2\eps)^2 > 1$ and hence partial recovery in the broadcast tree model is possible, and yet there is a monotone cutting adversary for which partial recovery in the semirandom broadcast tree model is information-theoretically impossible. 
\end{theorem}

Furthermore we analyze the recursive majority algorithm and show that it is robust to an even more powerful class of monotone adversaries, which are allowed to entirely control the subtree at a node whose label is different than its parent. We call this a strong adversary, and we prove:

\begin{theorem}\label{thm:intro-recmaj}
If $\frac{k}{\log k} (1 - 2\eps)^2 > 1 + o(1)$ and $\eps < \frac12$ then partial recovery in the broadcast tree model is possible, even with respect to a strong monotone adversary, where `$o(1)$' is taken as $k \to \infty$.
\end{theorem}

\noindent These results highlight another well-studied model where the introduction of a monotone adversary strictly changes what is possible. Nevertheless there is an algorithm that succeeds across the entire range of distributions that arise from the action of a monotone adversary, simultaneously. Interestingly, our robustness results can also be seen as a justification for why practitioners use recursive majority at all. It has been known for some time that recursive majority does not achieve the Kesten--Stigum bound \cite{mossel-recursive}
--- the threshold for reconstruction in the broadcast tree model --- although taking the majority vote of the leaves does. The advantage of recursive majority is that it is robust to very powerful adversaries while majority is not, and this only becomes clear when studying these algorithms through semirandom models!

\section{Models and Adversaries}\label{sec:models}

\subsection{Preliminaries}\label{sec:prelim}

Here we formally define the models we will be interested in, as well as the notion of partial recovery. Recall that $G(n,a/n,b/n)$ denotes the stochastic block model on two communities with $p_1 = p_2 = 1/2$ so that the communities are roughly equal sized. We will encode community membership as a \emph{label} $\sigma_v \in \{+1, -1\}$ on each node $v$. We will also refer to this as a \emph{spin}, following the convention in statistical physics. This numeric representation has the advantage that we can `add' spins in order to take the majority vote, and `multiply' them to compute the relative spin between a pair of nodes. We will be interested in the sparse regime $a, b = \Theta(1)$ where the graph has constant average degree, and we will assume $a > b$.

Next, we formally define partial recovery in the stochastic block model. Throughout this paper will will be interested in how our algorithms perform as $n$ (number of nodes) goes to $\infty$. We say that an event holds a.a.s.\ (asymptotically almost surely) if it holds with probability $1 - o(1)$ as $n \to \infty$. Similarly, we say that an event happens for a.a.e.\ (asymptotically almost every) $x$ if it holds with probability $1-o(1)$ over a random choice of $x$.

\begin{definition} We say that an assignment of $\{+1,-1\}$ spins to the nodes achieves {\bf $\eta$-partial recovery} if at least $\eta n$ of these spins match the true spins, or at least $\eta n$ match after a global flip of all the spins. Moreover, an algorithm that outputs a vector of spins $\hat \sigma \in \{+1,-1\}^n$ (indexed by nodes) is said to \emph{achieve partial recovery} and there exists $\eta > \frac{1}{2}$ such that $\hat \sigma$ achieves $\eta$-partial recovery a.a.s.\ in the limit $n \to \infty$.
\end{definition}

Next, we define the broadcast tree model (which we introduced informally earlier). The broadcast tree model is a stochastic process that starts with a single root node $\rho$ at level $0$ whose spin $\sigma_\rho \in \{+1,-1\}$ is chosen uniformly at random. Each node in turn gives birth to $\Pois(a/2)$ same-spin children and $\Pois(b/2)$ opposite-spin children, where $\Pois(c)$ is the Poisson distribution with expectation $c$. This process continues until level $R$ at which point it stops, and the nodes at level $R$ are called the \emph{leaves}. (The nodes on level $\le R-1$ that by chance do not give birth to any children are not considered leaves, even though they are leaves in the graph-theoretic sense.) An algorithm observes the spins at the leaf nodes and the topology of the tree, and the goal is to recover the root spin:

\begin{definition}
An algorithm that outputs a spin $\hat \sigma_\rho$ is said to \emph{achieve partial recovery} on the tree if there exists $\eta > \frac{1}{2}$ such that $\hat \sigma_\rho = \sigma_\rho$ with probability at least $\eta$, as $R \to \infty$. 
\end{definition}

The reparameterization in terms of $(k,\eps)$ becomes particularly convenient here: each node gives birth to $\Pois(k)$ children, so $k = \frac{a+b}{2}$ is the average branching factor. Moreover, each child has probability $\eps = \frac{b}{a+b}$ of having spin opposite to that of its parent.

It is known that taking the majority vote of the leaves is optimal in theory, in the sense that it achieves partial recovery for $k (1-2\eps)^2 > 1$ \cite{kesten-stigum} and that for $k (1-2\eps)^2 \leq 1$, partial recovery is information-theoretically impossible \cite{evans}. This is called the Kesten--Stigum bound, and it can also be interpreted as a condition on the second-largest eigenvalue of an appropriately defined transition matrix. There are many other natural variants of the broadcast tree model, that are more general instances of multi-type branching processes. However, even for simple extensions, the precise information-theoretic threshold is still unknown. In the Potts model, where nodes are labeled with one of $q$ labels, Sly \cite{sly-potts} showed that the Kesten--Stigum bound is not tight, as predicted by M\'ezard and Montanari \cite{mezard-montanari-reconstruction}. And for an asymmetric extension of the binary model above, Borgs et al.\ \cite{borgs-asymmetric} showed that the Kesten--Stigum bound is tight for some settings of the parameters. In our setting, this historical context presents a substantial complication because if we apply a monotone adversary to a broadcast tree model and it results in a complex propagation rule, there may not be good tools to prove that partial recovery is impossible.

\subsection{Discussion}

There is a close connection between the stochastic block model and the broadcast tree model, since the local neighborhood of any vertex in the graph is locally tree-like. Hence, if our goal is to prove a random vs.\ semirandom gap for community detection, a natural starting point is to establish such a gap for the broadcast tree model. It turns out that there is more than one natural way to define a monotone adversary for the broadcast tree model, but it will not be too difficult to establish lower bounds for either of them. The more difficult task is in finding an adversary that can plausibly be coupled to a corresponding adversary in the stochastic block model, and this will require us to put many sorts of constraints on the type of adversary that we should use to obtain a separation for the broadcast tree model. 

In the broadcast tree model, we will work with two notions of a monotone adversary. One is weak and will be used to show our separation results:

\vspace{1em}\fbox{\parbox{0.9\textwidth}{
{\bf Cutting Semirandom Model for Broadcast Tree:} 
\begin{itemize}
\item Sample a `precursor' broadcast tree $\Tpre$ from the broadcast tree model.
\item The monotone cutting `adversary' can delete any number of edges between nodes of different labels, thus removing subtrees.
\item Output the resulting tree $T$. (The removed subtrees are not revealed.)
\end{itemize}
}}
\vspace{1em}

\noindent Our other adversary is stronger, and we will establish upper bounds against this adversary (with the recursive majority algorithm) to give a strong recoverability guarantee:

\vspace{1em}\fbox{\parbox{0.9\textwidth}{
{\bf Strong Semirandom Model for Broadcast Tree:} 
\begin{itemize}
\item Sample a `precursor' broadcast tree $\Tpre$ from the broadcast tree model.
\item Whenever a child has the opposite label from its parent, the strong monotone `adversary' can replace the entire subtree, starting from the child, with a different subtree (topology and labels) of its choice.
\item Output the resulting tree $T$.
\end{itemize}
}}
\vspace{1em}

\noindent An upper bound against this latter model amounts to a recovery guarantee without any assumptions as to what happens after a `mutation' of labels --- for example, a genetic mutation might affect reproductive fitness and change the birth rule for the topology. Other variants of monotone adversaries could also be justified.

\paragraph{Majority Fails} It is helpful to first see how adversaries in these models might break existing algorithms. Recall that in the broadcast tree model, taking the majority vote of the leaves yields an algorithm that works up to the Kesten--Stigum bound, and this is optimal since reconstruction is impossible beyond this. In the language of $k$ and $\eps$, each node gives birth to $\Pois(k(1-\eps))$ nodes of the same label and $\Pois(k\eps)$ nodes of the opposite label. Hence in a depth $R$ tree we expect $k^R$ leaves, but the total bias of the spins can be recursively computed as $k^R(1-2\eps)^R$. The fact that majority works can be proven by applying the second moment method and comparing the bias to its variance.

However, an overwhelming number of the leaves are on some path that has a flip in the label at some point: we only expect $k^R(1-\eps)^R$ nodes with all-root-spin ancestry, a vanishing proportion as $R \to \infty$. The strong monotone adversary has control over all the rest, and can easily break the majority vote. Even the monotone cutting adversary can control the majority vote, by cutting leaves whose spin matches the root but whose parents have the opposite label. This happens for a constant fraction of the leaf nodes, and this change overwhelms the majority vote. So majority vote fails against the semirandom model, for \emph{all} nontrivial $k$ and $\eps$.

This is an instructive example, but we emphasize that breaking one algorithm does not yield a lower bound. For example, if the algorithm knew what the adversary were doing, it could potentially infer information about where the adversary has cut edges based on the degree profile, and could use this to guess the label of the root.

\paragraph{The Problem of Orientation} Many first attempts at a separation in the broadcast tree model (which work!)\ rely on knowing the label of the root. However, such adversaries present a major difficulty in coupling them to a corresponding adversary in the graph. Each graph gives rise to many overlapping broadcast trees (the local neighborhood of each vertex) and a graph adversary needs to simultaneously make all of these tree reconstruction problems harder. This means a graph adversary cannot focus on trying to hide the spin of a specific tree root; rather, it should act in a local, orientation-free way that inhibits the propagation of information in all directions.

A promising approach is to look for nodes $v$ whose neighbors in $\Gpre$ all have the opposite label, and cut all of these edges. Such nodes serve only to further inter-connect each community, and cutting their edges would seem to make community detection strictly harder. In the corresponding broadcast tree model, these nodes $v$ represent flips in the label that are entirely corrected back, and deleting them penalizes any sort of over-reliance on distributional flukes in how errors propagate. For example, majority reconstruction in the tree fully relies on predictable tail events whereby nodes with label different from the root may lead to subtrees voting in the correct direction nonetheless. 

\paragraph{The Problem of Dependence} Now, however, a different sort of problem arises: if we were to naively apply the adversary described above to a broadcast tree, this would introduce complicated distance-$2$ dependencies in the distribution of observed spins, as certain diameter-$2$ spin patterns are banned in the observed tree (as they would have been cut). In particular, the resulting distribution is no longer a Markov random field. This is not inherently a problem, in that we could still hope to prove stronger lower bounds for such models beyond the Kesten--Stigum bound. However, the difficulty is that even for quite simple models on a tree (e.g.\ the Potts model \cite{sly-potts}, asymmetric binary channels \cite{borgs-asymmetric}) the threshold is not known, and the lower bound techniques that establish the Kesten--Stigum bound seem to break down.

An alternative is to specifically look for degree-$2$ nodes $v$ whose neighbors in $\Gpre$ have the opposite label, and cut both incident edges. Although there are still issues about making this a Markov random field, we can alleviate them by adding a $3$-clique potential on each degree-$2$ node and its two neighbors. Then after we marginalize out the label of the degree-$2$ node, the $3$-clique potential becomes a $2$-clique potential, and we return to a Markov random field over a tree! In other words, if we ignore the spin on a degree-2 node and treat its two incident edges like a single edge, we return to a well-behaved spin propagation rule.

\subsection{Our Adversary}

We are now ready to describe the adversary that we will use to prove Theorems~\ref{thm:intro-main} and~\ref{thm:intro-treehard}. We only need two additional adjustments beyond the discussion above. Instead of making every possible degree-$2$ cutting move as described earlier, we will only cut with probability $\delta$. We will tune $\delta$ to ensure that the monotone changes we make do not overshoot and accidentally reveal more information about the underlying communities by cutting in too predictable a fashion. Finally, our adversary adds local restrictions to where and how it cuts, to ensure that the changes it makes do not overlap or interfere with each other (e.g.\ chains of degree-$2$ nodes). These details will not be especially relevant until Section~\ref{sec:no-lrc}, where they simplify the combinatorics of guessing the original precursor graph from the observed graph.

\begin{distribution}\label{dist:graph}

Let $(a,b,n)$ be given. Write $k = \frac{a+b}{2}$ and $\eps = \frac{b}{a+b}$. We sample a `precursor' graph $\Gpre \sim G(n,a/n,b/n)$, and apply the following adversary:

\vspace{1em}\fbox{\parbox{0.94\textwidth}{
{\bf Adversary:}
\begin{itemize}
\item If at least $3$ neighbors of a vertex $v$ have degree not equal to $2$, mark $v$ \good.
\item If a degree-$2$ vertex has both of its neighbors marked \good, mark it \marked.
\item For each \marked node $v$ whose two neighbors $w_1$, $w_2$ both have opposite label to $v$: with probability $\delta$, delete both edges $(v,w_1)$ and $(v,w_2)$ (otherwise keep both edges) where
\begin{equation}\label{eq:delta}
\delta \defeq \begin{cases}
1 & \text{if $\eps \leq \frac13$}, \\
\frac{(1-2\eps)^2}{\eps^2} & \text{if $\eps \geq \frac13$.}\end{cases}
\end{equation}
\end{itemize}
}}
\end{distribution}

We can now outline the proof of our main theorem. In order to show that partial recovery is impossible, it suffices to show that it is impossible to reconstruct the relative spin (same or different) of two random nodes $u$ and $v$, better than random guessing.
Before applying the adversary, an $O(\log n)$-radius neighborhood $U$ around $u$ resembles the broadcast tree model rooted at $u$. After the adversary is applied to the graph, $U$ resembles a broadcast tree model with a corresponding cutting adversary applied. This resemblance will be made precise by the coupling argument of Section~\ref{sec:coupling}; there will be some complications in this, and our tree will not be uniformly the same depth but will have a serrated boundary.

We show in Section~\ref{sec:treehard} that when the cutting adversary is applied to the tree, the tree reconstruction problem (guess the spin of $u$ from the spins on the boundary of $U$) becomes strictly harder: the average branching factor becomes lower due to sufficient cutting, while the new spin propagation rule resembles classical noisy propagation with at least as much noise (so long as we marginalize out the spins of \marked nodes). We then apply the proof technique of Evans et al.\ \cite{evans} to complete the tree lower bound.

The final step in the proof is to show that reconstructing the relative spin of $u$ and $v$ is at least as hard as reconstructing the spin of $u$ given the spins on the boundary of $U$ (which separates $u$ and $v$ with high probability). This is one of the most technically involved steps in the proof. In the lower bound for the random model \cite{mns}, this step --- called ``no long-range correlations'' --- was already an involved calculation on a closed-form expression for $\prob{G \given \sigma}$. In our setting, in order to get a closed-form expression for the conditional probability, we will sum over all the possible precursor graphs $\Gpre$ that could have yielded the observed graph $G$. We characterize these precursors in Lemma~\ref{lemma:unsurgery}, and the main reason for the \good and \marked nodes in our adversary construction is to simplify this process.

A natural open question is whether one can find the optimal monotone adversary. For instance, we have not even used the power to add edges within communities. Note however, that our current adversary is delicately constructed in order to make each step of the proof tractable. It is not too hard to propose alternative adversaries that seem stronger, but it is likely that one of the major steps in the proof (tree lower bound or no long-range correlations) will become prohibitively complicated.
Recent predictions on the performance of SDP methods \cite{jm-phase} could be suggestive of the true semirandom threshold.

\section{Coupling of Graph and Tree Models}
\label{sec:coupling}

It is well-known that sparse, random graphs are locally tree-like with very few short cycles. This is the basis of Mossel--Neeman--Sly's approach in coupling between a local neighborhood of the stochastic block model and the broadcast tree model \cite{mns}. Hence, our first order of business will be to couple a typical neighborhood from our graph distribution (Distribution~\ref{dist:graph}) to the following tree model, against which we can hope to show lower bounds.

\begin{distribution} Given the parameters $(a,b,R)$, generate a tree $T$ with spins $\sigma$ as follows:
\label{dist:tree}
\begin{itemize}
\item Start with a root vertex $\rho$, with spin $\sigma_\rho$ chosen uniformly at random from $\pm 1$.
\item Until a uniform depth $R+3$ (where the root is considered depth 0), let each vertex $v$ give birth to $\Pois(a/2)$ children of the same spin and $\Pois(b/2)$ children of opposite spin.
\item Apply the graph adversary (from Distribution~\ref{dist:graph}) to this tree; this involves assigning markings \good and \marked, and cutting some \marked nodes. Keep only the connected component of the root.
\item Remove all nodes of depth greater than $R$ (the bottom $3$ levels).
\item Remove any \marked node at depth $R$ along with its siblings, exposing the parent as a leaf.
\end{itemize}
\end{distribution}
The reason we trim $3$ levels at the bottom is to ensure that the markings and cuttings in $T$ match those in $G$, since these depend on the radius-$3$ neighborhood of each node and edge, respectively. Removing \marked nodes at level $R$ will ensure that the more complicated spin interactions of \marked nodes and their neighbors do not span across the boundary of a tree neighborhood in the graph --- we want to cleanly separate out a tree recovery problem. We use a slightly non-conventional definition of `leaves':
\begin{definition}
\label{def:leaves}
When we refer to the {\bf leaves} of a tree sampled from Distribution~\ref{dist:tree} we mean the nodes at depth $R$ plus any nodes at depth $R-1$ that are revealed during the last step. Nodes at depth $\le R-1$ that happen to give birth to no children are not considered leaves; if the root is \marked and gets cut by the adversary so that the tree is a single node, this single node is not considered a leaf.
\end{definition}

We can couple the above tree model to neighborhoods of the graph:
\begin{proposition}\label{prop:coupling}
Let $(a,b,n)$ be given, and let $\rho_G$ be any vertex of the graph. Let 
\begin{equation}\label{eq:R} R = \left\lfloor\frac{\log n}{10 \log(2(a+b))} \right\rfloor - 3. \end{equation}
There exists a joint distribution $(G,\sigma_G,T,\sigma_T)$ such that the marginals on $(G,\sigma_G)$ and $(T,\sigma_T)$ match the graph and tree models, respectively, while with probability $1 - o(1)$:
\begin{itemize}
\item there exists an isomorphism (preserving edges, spins, and the markings \good and \marked) between the tree $T$ and a vertex-subset $U$ of $G$,
\item the tree root corresponds to the vertex $\rho_G$ in the graph, and the leaves correspond to a vertex set $B$ that separates the interior $U \setminus B$ from the rest of the graph $G \setminus U$.
\item we have $|U| = O(n^{1/8})$.
\end{itemize}
\end{proposition}
\begin{proof}
As proven in \cite{mns}, there exists a coupling between the precursor graph $(\Gpre,\sigma_G)$ and the precursor tree $(\Tpre,\sigma_T)$, such that $\Tpre$ matches the radius $R+3$ neighborhood $U$ of $\rho_G$ in $\Gpre$, which has size $O(n^{1/8})$; this fails with probability $o(1)$.

Next the adversary assigns vertex markings (\good and \marked) to both models. The marking of each vertex is deterministic and only depends on the topology of the radius-$3$ neighborhood of the vertex. Thus the markings in $\Gpre$ and $\Tpre$ match up to radius $R$. Some \marked nodes are \emph{cuttable}, i.e.\ both their neighbors have opposite spin; the cuttable nodes in $\Gpre$ and $\Tpre$ also match up to radius $R$. The adversary now cuts the edges incident to a random subset of cuttable vertices; we can trivially couple the random choices made on $\Gpre$ with those on $\Tpre$ up to radius $R$. We keep only the connected component of the root in $T$; likewise let us keep only the corresponding vertices in $U$, i.e.\ only those still connected to $\rho_G$ by a path in $U$.

After removing nodes of depth greater than $R$, we have removed the subset of $T$ for which the markings and the action of the adversary differ from those in $G$. Thus at this stage, the tree exactly matches the radius $R$ neighborhood of $\rho_G$ in $G$, along the isomorphism given by the coupling from \cite{mns}. Any boundary vertex of this neighborhood must have distance exactly $R$ from $\rho_G$, thus its corresponding vertex in the tree has depth $R$ and is a leaf. Passing to the final step of removing \marked leaves in $T$ and their siblings, we remove their corresponding nodes from $U$; this does not change the boundary-to-leaves correspondence.
\end{proof}

\section{Random vs.\ Semirandom Separation in the Tree Model}\label{sec:treehard}

In this section we show that our tree distribution (Distribution~\ref{dist:tree}) evidences a random vs.\ semirandom separation in the broadcast tree model. Recall that the goal is to recover the spin of the root from the spins on the leaves, given full knowledge of the tree topology and the node markings (\good and \marked). Recall the non-conventional definition of `leaves' (Definition~\ref{def:leaves}).

Let $\Delta = \Delta(a,b,R)$ be the advantage over random guessing, defined such that $\frac{1+\Delta}{2}$ is the probability that the optimal estimator (maximum likelihood) succeeds at the above task. We will show that our tree model is asymptotically infeasible $(\Delta \to 0$ as $R \to \infty$) for a strictly larger range of parameters $(a,b)$ than that of the corresponding random model.

\begin{proposition}\label{prop:treehard} For every real number $k > 1$, there exists $0 < \eps < \frac12$ such that $k (1-2\eps)^2 > 1$, yet for a tree sampled from Distribution~\ref{dist:tree} with parameters $a = 2k(1-\eps)$, $b = 2k\eps$ and depth $R$, we have $\Delta(a,b,R) \to 0$ as $R \to \infty$.
\end{proposition}
Recall that the condition $k(1-2\eps)^2 > 1$ is the classical Kesten--Stigum bound, which is sufficient to beat random guessing in the random model \cite{kesten-stigum}. Several decades later, this bound was found to be tight for the random model \cite{evans}. There remain many open questions regarding the hardness of tree models, and some care was required in crafting an adversary for this problem that keeps the proof of this lower bound tractable.

Broadly, the Kesten--Stigum bound asserts that recoverability depends on the average branching factor (a contribution from the tree topology) and the amount of noise (a contribution from the spin propagation rule). The first step of our proof is to distinguish these in our distribution: we can first generate a tree topology from the appropriate marginal distribution, and then sample spins from the conditional distribution given this tree. We will show how to view this distribution on spins within the lower bound framework of \cite{evans}. Moreover, the new spin propagation rule is at least as noisy as the original, while our cutting adversary has strictly decreased the average branching factor.

\begin{proof}[Proof of Proposition~\ref{prop:treehard}]
We first re-state our tree model in terms of topology generation followed by spin propagation. Instead of letting each node give birth to $\Pois(a/2)$ same-spin children and $\Pois(b/2)$ opposite-spin children, we can equivalently let each node $v$ give birth to $\Pois(k)$ children and then independently choose the spin of each child to be the same as $v$ with probability $1-\eps$ and opposite to $v$ otherwise. (The equivalence of these steps is often known as `Poissonization'.) Here the correspondence between $(a,b)$ and $(k,\eps)$ is as usual: $k = \frac{a+b}{2}$ and $\eps = \frac{b}{a+b}$. This allows us to first sample the entire tree topology without spins. Then we can add markings (\good and \marked), since these depend only on the topology. Next, we sample the spins as above, by generating an appropriate independent $\pm 1$ value $f_e$ on each edge, indicating whether or not a sign flip occurs across that edge. Finally, we cut edges according to the adversary's rule.

\begin{distribution} Given the parameters $(a,b,R)$, generate a tree $T$ with spins $\sigma$ as follows:
\label{dist:tree-1}
\begin{itemize}
\item Start with a root vertex $\rho$. Until a uniform depth $R+3$, let each vertex $v$ give birth to $\Pois(k)$ nodes.
\item Mark nodes as \good and \marked according to the rules of the graph adversary.
\item For each edge $e$, generate an independent flip value $f_e$ which is $+1$ with probability $1-\eps$ and $-1$ otherwise.
\item Choose the root spin $\sigma_\rho$ uniformly from $\pm 1$. Propagate spins down the tree, letting $\sigma_v = \sigma_u f_e$ where $e$ is the edge connecting $v$ to its parent $u$.
\item Cut edges according to the adversary's rule, keeping only the connected component of the root.
\item Trim the bottom of the tree (according to the last two steps of Distribution~\ref{dist:tree}).
\end{itemize}
\end{distribution}

It is clear that this tree distribution (Distribution~\ref{dist:tree-1}) is identical to the original tree distribution (Distribution~\ref{dist:tree}). Our next step will be to re-state this model in yet another equivalent way. The goal now is to sample the final tree topology (including which edges get cut) before sampling the spins. Consider a \marked node $v$, its parent $u$, and its single child $w$. Instead of writing the spin propagation as independent flips $f_{(u,v)}$ and $f_{(v,w)}$, we will write it as random variables $c_v$ and $f_v$. Here $c_v$ is equal to 1 if the adversary decides to cut $v$ (and 0 otherwise), and $f_v$ is equal to 1 if $\sigma_u = \sigma_w$ (and $-1$ otherwise). This means if $f_{(u,v)} = f_{(v,w)} = 1$ then $c_v$ is 1 with probability $\delta$ (and 0 otherwise); and if we do not have $f_{(u,v)} = f_{(v,w)} = 1$ then $c_v = 0$. Hence $c_v = 1$ with probability $\eps^2 \delta$. If $c_v = 1$ then $f_v$ is irrelevant because the adversary will cut $v$ from the tree. Conditioned on $c_v = 0$, $f_v$ takes the value $+1$ with probability $\frac{(1-\eps)^2+\eps^2(1-\delta)}{1-\eps^2 \delta}$ and $-1$ otherwise. This means that for a \marked (but not cut) node $v$, the joint distribution of $(\sigma_u,\sigma_w)$ obeys a propagation rule that is equivalent to putting noise $\eps'$ (instead of $\eps$) on edges $(u,v)$ and $(v,w)$, where (using the definition of $\delta$)
$$ \eps' = \begin{cases}
\frac12 - \frac12 \sqrt{\frac{1-3\eps}{1+\eps}} & \text{if $\eps \leq \frac13$,} \\
\frac12 & \text{if $\eps \geq \frac13$.}
\end{cases}$$
One can verify that $\eps < \eps' \le \frac{1}{2}$ for all $\eps \in (0,\frac12)$. For most \marked nodes in the tree, we are simply going to replace $\eps$ by $\eps'$ on the incident edges, which gives the correct joint distribution of $(\sigma_u,\sigma_w)$ but not the correct joint distribution of $(\sigma_u,\sigma_v,\sigma_w)$. This is acceptable because the distribution of leaf spins (given the root spin) is still correct; for this reason we have made sure that none of the leaves are \marked nodes. The only time when we actually care about the spin of a \marked node is in the case where the root $\rho$ is \marked. In this case, the root might be cut by the adversary, yielding a 1-node tree (with no revealed leaves). Otherwise, if the root is \marked but not cut, our spin propagation rule needs to sample the spins of the root's two children (a \marked node that is not cut must have degree $2$) from the appropriate joint distribution over $\{\pm 1\}^2$; let $\mathcal{D}^+$ denote this distribution, conditioned on $\sigma_\rho = +1$. It will not be important to compute $\mathcal{D}^+$ explicitly (although it is straightforward to do so). We are now ready to state the next tree model. This model is equivalent to the previous ones in that the joint distribution of the root spin, topology, markings, and leaf spins is the same. The spins on the \marked nodes (other than the root) do not have the same distribution as before, but this is irrelevant to the question of recovering the root from the leaves.

\begin{distribution} Given the parameters $(a,b,R)$, generate a tree $T$ with spins $\sigma$ as follows:
\label{dist:tree-2}
\begin{itemize}
\item Start with a root vertex $\rho$. Until a uniform depth $R+3$, let each vertex $v$ give birth to $\Pois(k)$ nodes.
\item Mark nodes as \good and \marked according to the rules of the graph adversary.
\item Decide which nodes the adversary should cut: cut each \marked node independently with probability $\eps^2 \delta$.
\item Let $\eps_e = \eps'$ for edges $e$ that are incident to a \marked node, and let $\eps_e = \eps$ for all other edges.
\item For each edge $e$, generate an independent flip value $f_e$ which is $+1$ with probability $1-\eps_e$ and $-1$ otherwise.
\item Choose the root spin $\sigma_\rho$ uniformly from $\pm 1$. If the root is \marked but not isolated: let $u_1,u_2$ be its children, draw $(d_1,d_2) \sim \mathcal{D}^+$, and let $\sigma_{u_1} = d_1 \sigma_\rho$, $\sigma_{u_2} = d_2 \sigma_\rho$. For all other nodes, propagate spins down the tree as usual, letting $\sigma_v = \sigma_u f_e$ where $e$ is the edge connecting $v$ to its parent $u$.
\item Trim the bottom of the tree (according to the last two steps of Distribution~\ref{dist:tree}).
\end{itemize}
\end{distribution}

Our next step is to further modify this tree model in ways that only make it easier, in the sense that the advantage $\Delta(a,b,R)$ can only increase. First, we will address the issue of the complicated propagation rule $\mathcal{D}^+$ in the case that the root is \marked. Suppose the root is \marked (but not cut) and consider deterministically setting $\sigma_{u_1} = \sigma_{u_2} = \sigma_\rho$ where $u_1,u_2$ are the root's two children. From there, spin propagation continues as usual. We claim that this new model can only be easier than the original one. To see this, note that the new model can `simulate' the original one: upon observing leaf spins drawn from the new model, drawn $(d_1, d_2) \sim \mathcal{D}^+$ and then, for each $i \in \{1,2\}$ and for each leaf $w$ descended from $u_i$, replace $\sigma_w$ by $d_i \sigma_w$. For convenience we will also replace the first level of the tree by deterministic zero-noise propagation in the case where the root is not \marked. We can similarly argue that the model only gets easier, since one can simulate the old model using the new model by sampling the noise on the first level.

Note that we now have a tree model such that once the topology is chosen, the spin-propagation rule is very simple: at each edge $e$, a sign flip occurs independently with some probability $\eps_e \in \{\eps,\eps',0\}$. Hardness results for such a tree model were studied by Evans et al., who established the following bound on the advantage $\Delta_T$ for a fixed tree topology $T$ (\cite{evans} Theorem 1.3'):
$$ \Delta_T^2 \le 2 \sum_{\text{leaves }v} \Theta_v^2 \quad\text{ where }\quad \Theta_v \defeq \prod_{e \in \mathrm{path}(v)} \theta_e, $$
and where $\theta_e \defeq 1-2\eps_e$, and $\mathrm{path}(v)$ denotes the unique path from the root to leaf $v$. In our case, the tree $T$ (including both the tree topology and markings) is random, so the advantage is $\Delta = \EE_T[\Delta_T]$, by the law of total probability. By Jensen's inequality,
$$\Delta^2 = \left(\EE_T[\Delta_T]\right)^2 \le \EE_T[\Delta_T^2] \le \EE_T\left[ 2 \sum_{\text{leaves }v} \Theta_v^2 \right] \defeq W.$$
Our goal is to show $\Delta \to 0$ (as $R \to \infty$), so it is sufficient to show $W \to 0$. Any leaf in our tree is at height $\ge R-1$. Furthermore we have $\eps_e = 0$ for edges incident to the root and $\eps \le \eps_e \leq \frac12$ elsewhere. Therefore, for any leaf $v$, we can bound $\Theta_v \le (1-2 \eps)^{R-2}$. Note that we have not actually used the fact that $\eps'$ strictly exceeds $\eps$; we will in fact be able to prove our result by leveraging only the decrease in branching factor and not the noise increase. Now we have
$$W = \EE_T\left[ 2 \sum_{\text{leaves }v} \Theta_v^2 \right] \le 2 \, \EE_T[\text{\# leaves}] (1-2 \eps)^{2(R-2)}$$
and so we need only to bound the expected number of leaves. For $i = 0, \ldots, R$, let \emph{level $i$} denote the set of vertices at distance $i$ from the root (so the root is on level 0). Call levels $1,7, \ldots, 6j+1, \ldots$ the \emph{base} levels and call levels $4,10, \ldots, 6j-2, \ldots$ the \emph{cutting} levels. Imagine growing the tree using $\Pois(k)$ births as usual, and marking nodes as \good and \marked as usual. Now allow the adversary to cut vertices as usual, except refuse to cut any \marked node that is not on a cutting level. Note that this can only result in more leaves, and so this new process will give an upper bound on the expected number of leaves in our actual model.

To analyze the new process, suppose we have a node $v$ on a base level. The number of descendants $n_v$ of $v$ on the next base level is a random variable whose distribution does not depend on the ancestors of $v$; this independence is the reason for defining base levels and cutting levels. Define $k'$ by $(k')^6 \defeq \EE[n_v]$; then $(k')^6 < k^6$, since there is some nonzero probability that one of $v$'s descendants will be cut at the subsequent cutting level. The expected number of leaves in the entire tree is now at most $k^6 (k')^{R-6}$, where the factor $k^6$ accounts for the first level and the last $\le 5$ levels that are not followed by a base level. Now $W \le 2 k^6 (k')^{R-6} (1-2 \eps)^{2(R-2)}$, which goes to 0 as $R \to \infty$ provided that
$$k'(1-2 \eps)^2 < 1. $$
Hence this inequality suffices for impossibility of recovery.

However, $k'$ depends on the topology of the final tree model, which depends on $k$, $\eps$, and $\delta$, so this inequality is slightly more complicated than it looks. We write $k'(k,\eps)$ to make this dependence explicit; recall that $\delta$ is a (continuous) function of $\eps$. We argued above that $k'(k,\eps) < k$ for all $0 \le \eps \le \frac12$. In particular, at the critical value $\epscrit = \frac12 ( 1 - \frac{1}{\sqrt{k}})$ for the \emph{random} model, we have
$$k'(k,\epscrit) (1-2\epscrit)^2 < k (1-2\epscrit)^2 = 1.$$
But $k'(k,\eps) (1-2\eps)^2$ is a continuous function of $\eps$, so if we take $\eps$ to be slightly less than $\epscrit$, we must still have $k'(k,\eps) (1-2\eps)^2 < 1$, while $k(1-2\eps)^2 > 1$, as desired.
\end{proof}

This completes the proof of Theorem~\ref{thm:intro-treehard} (semirandom vs.\ random separation on trees). In Appendix~\ref{sec:lower-bound-explicit}, we explicitly compute a lower bound on the separation $k^6 - k'(k,\eps)^6$.

\section{No Long-Range Correlations}\label{sec:no-lrc}

The lower bound of the previous section will form the core of a lower bound for the graph: we have already established that, for a large range of parameters, it is impossible to learn the spin of a node $u$ purely from the spins at the boundary $B$ of its tree-like local neighborhood. We will now see why this makes recovery in the graph impossible: there is almost nothing else to learn about $\sigma_u$ from beyond its local neighborhood once we know the spins on $B$.

\begin{lemma}[No long-range correlations]\label{lemma:no-lrc}
Let a graph $G$ (including markings \good and \marked) and spins $\sigma$ be drawn from Distribution~\ref{dist:graph}.
Let $A = A(G)$, $B = B(G)$, $C = C(G)$ be a vertex-partition of $G$ such that
\begin{itemize}
\item $B$ separates $A$ from $C$ in $G$ (no edges between $A$ and $C$),
\item $|A \cup B| = O(n^{1/8})$,
\item $B$ contains no \marked nodes.
\end{itemize}
Then
$ \prob{\sigma_A \given G, \sigma_{BC}} = (1 + o(1)) \prob{\sigma_A \given G, \sigma_B}, $
for asymptotically almost every (a.a.e.) $G$ and $\sigma$.
\end{lemma}

\noindent Here $\sigma_A$ denotes the spins on $A$ and $\sigma_{BC}$ denotes the spins on $B \cup C$.

To clarify, when we refer to \good or \marked nodes in $G$ we are referring to the original markings that they were assigned in $\Gpre$. For instance, if a \marked node is cut by the adversary, it is still considered \marked in $G$ even though it no longer has degree $2$. Recall that the markings (\good and \marked) in $G$ are revealed to the reconstruction algorithm.

In Lemma~\ref{lemma:no-lrc}, we do crucially use that $B$ does not contain \marked nodes: if a node in $B$ were \marked with spin $+1$, and we then revealed that its neighbor in $C$ has spin $-1$, this would strengthen our belief that the neighbor in $A$ has spin $+1$, as otherwise the \marked node would have some substantial probability of having been cut, which is observed not to be the case. So Lemma~\ref{lemma:no-lrc} would be false if we allowed \marked nodes in $B$.

\subsection{Structure of Possible Precursors $\Gpre$}

The proof of Lemma~\ref{lemma:no-lrc} will require a thorough understanding of the distribution of spins given $G$, which we can only obtain by understanding the possible precursors of $G$ under the adversary. The reason for the \good and \marked nodes in our adversary construction is to make these precursors well-behaved. We start with some simple observations. Let $\Gpre$ be a graph that yields $G$ (with nonzero probability) by application of the adversary.
\begin{observation} \label{obs:good}
A \good node has degree at least $3$ in $G$.
\end{observation}
\begin{proof}
A \good node has at least three non-degree-$2$ neighbors in $\Gpre$ and the adversary will not remove these.
\end{proof}
\begin{observation} \label{obs:create}
The adversary does not create any new degree-$2$ nodes. In other words, if a node has degree $2$ in $G$ then it has degree $2$ in $\Gpre$ (with the same two neighbors).
\end{observation}
\begin{proof}
In order to create a degree-$2$ node $v$, the adversary must cut at least one edge incident to $v$. This can only happen if $v$ is either \marked or \good. It cannot be \marked because it does not have degree $2$ in $\Gpre$. But it also cannot be \good or else it has degree at least 3 in $G$ (by Observation~\ref{obs:good}).
\end{proof}
\noindent The key property of the \good/\marked construction is that the adversary cannot change whether or not a node has the following `goodness' property.
\begin{definition}
Say that a node has the {\bf goodness} property with respect to a graph if at least three of its neighbors do not have degree $2$.
\end{definition}
\begin{lemma}
\label{lemma:goodness}
The \good nodes are precisely the nodes that have the goodness property with respect to $G$.
\end{lemma}
\noindent Note that this is not tautological because the \good nodes are defined as the nodes that have the goodness property with respect to $\Gpre$, not $G$.
\begin{proof}
We need to show that the goodness property is invariant under the adversary's action. First we assume $v$ has goodness in $\Gpre$ and show that it also has goodness in $G$. Since $v$ has the goodness property in $\Gpre$, it has three neighbors $u_1,u_2,u_3$ in $\Gpre$ that do not have degree $2$. Each $u_i$ remains connected to $v$ in $G$ since the adversary only cuts edges that are incident to a degree-$2$ node. Furthermore, each $u_i$ does not have degree $2$ in $G$ because degree-$2$ nodes cannot be created (Observation~\ref{obs:create}).

Now we show the converse: assume $v$ does not have goodness in $\Gpre$ and show that it still does not have goodness in $G$. The only way that $v$ could obtain goodness is if at least one of its degree-$2$ neighbors $u$ becomes non-degree-$2$ (while remaining connected to $v$). But this is impossible because whenever the adversary cuts an edge incident to a degree-2 vertex $u$, it causes $u$ to become isolated in $G$.
\end{proof}

\noindent Next we state an easy fact about the structure of \marked nodes in $G$.
\begin{lemma}
\label{lemma:marked}
The \marked nodes in $G$ are precisely the degree-$2$ nodes with two \good neighbors, plus some isolated nodes.
\end{lemma}
\begin{proof}
Every \marked node in $\Gpre$ has degree-$2$ with two \good neighbors. If the node is cut, it becomes isolated in $G$; otherwise it remains degree-$2$ with two \good neighbors. Conversely, let $v$ be degree-$2$ in $G$ with two \good neighbors; we will show $v$ is \marked. Since degree-$2$ nodes cannot be created (Observation~\ref{obs:create}), $v$ has degree-$2$ in $\Gpre$ with the same two \good neighbors, and is therefore \marked.
\end{proof}

\noindent Now we are ready to characterize the possible precursors $\Gpre$ of a given $G$.
\begin{lemma}\label{lemma:unsurgery}
Suppose we have a graph $G$ (including node markings) and spins $\sigma$, drawn from Distribution~\ref{dist:graph}. The probability that a graph $\Gpre$ yields $G$ under the action of the adversary is zero unless $\Gpre$ can be obtained from $G$ by connecting each isolated \marked node $v$ of $G$ to exactly two \good nodes of opposite spin to $v$. In this case, if $G$ has $m = m(G)$ isolated \marked nodes and $w = w(\sigma,G)$ \marked nodes with two opposite-spin neighbors, then
$$ \prob{G \given \Gpre, \sigma} = (1-\delta)^w \delta^m $$
where we define $0^0 = 1$ in the case $\delta = 1, w = 0$.
\end{lemma}
\noindent Recall that $\delta$ is the probability with which the adversary cuts each $\marked$ node that has two opposite-spin neighbors.
\begin{proof}
First suppose $\Gpre$ can yield $G$ via the adversary. We will show that $\Gpre$ takes the desired form. The nodes cut by the adversary are precisely the isolated \marked nodes of $G$. Every such node was originally connected to two opposite-spin \good nodes in $\Gpre$. Therefore $\Gpre$ can be obtained from $G$ by connecting each isolated \marked node to exactly two opposite-spin \good nodes. Conversely, let $\Gpre$ be obtained from $G$ by connecting each isolated \marked node to exactly two opposite-spin \good nodes. The \good nodes in $G$ are (by Lemma~\ref{lemma:goodness}) precisely the nodes that have the goodness property in $G$. These are also precisely the nodes that have the goodness property in $\Gpre$, since the process of connecting each isolated \marked node to two \good nodes does not change goodness. Consider running the adversary on $\Gpre$. The nodes that it marks \good will be precisely the \good nodes in $G$. Also, the nodes that it marks \marked will be precisely the \marked nodes in $G$; this is clear from Lemma~\ref{lemma:marked}. This means the adversary will output $G$ iff it chooses to cut the $m$ \marked nodes that are isolated in $G$ and chooses not to cut the $w$ \marked nodes that have two opposite-spin neighbors in $G$. This happens with probability $(1-\delta)^w \delta^m$.
\end{proof}

\subsection{Proof of No Long-Range Correlations}

\begin{proof}[Proof of Lemma~\ref{lemma:no-lrc}]

We proceed as follows. In the ordinary stochastic block model, given an observed graph, the probability of any set of spins $\sigma$ factorizes as a product of pairwise interactions. In our model, by summing over all possible precursors $\Gpre$ that could have lead to $G$ via the adversary, we find the same pairwise interactions together with a further global, combinatorial interaction. We show that the neighborhood $A$ is too small to make a significant impact on this global interaction, while the pairwise interactions between $A$ and $C$ are weak, so that the only factors relevant to $A$ are the pairwise interactions within $A \cup B$, which are independent of the spins in $C$.

Let $\Gpre$ denote the graph before the action of the adversary. Then $\prob{\Gpre \given \sigma}$ factors into the following potentials on unordered pairs: $\prob{\Gpre \given \sigma} = \Phi(\sigma,\Gpre) \defeq \prod_{u,v} \phi(\sigma_u,\sigma_v,\Gpre)$ with $u,v$ ranging over unordered pairs of distinct vertices, where
$$ \phi(\sigma_u,\sigma_v,H) = \begin{cases}
a/n   & \text{if $u     \sim v$ and $\sigma_u  =   \sigma_v$,} \\
b/n   & \text{if $u     \sim v$ and $\sigma_u \neq \sigma_v$,} \\
1-a/n & \text{if $u \not\sim v$ and $\sigma_u  =   \sigma_v$,} \\
1-b/n & \text{if $u \not\sim v$ and $\sigma_u \neq \sigma_v$}
\end{cases} $$
where $\sim$ denotes adjacency in the graph $H$.

To leverage this description, let us sum over all possible precursors $\Gpre$ of $G$ under the adversary. Let $L(\sigma,G)$ denote the set of possible precursors $\Gpre$ of $G$, as described by Lemma~\ref{lemma:unsurgery}: $\Gpre$ is obtained from $G$ by connecting each of the $m$ isolated \marked nodes to exactly two \good nodes of the opposite spin.
\begin{align*}
\prob{\sigma \given G}
&= \sum_{\Gpre \in L(\sigma,G)} \prob{\Gpre \given \sigma} \prob{\sigma} \prob{G \given \sigma,\Gpre} / \prob{G} \\
&\propto \sum_{\Gpre \in L(\sigma,G)} \Phi(\sigma,\Gpre) \prob{G \given \sigma,\Gpre} \\
&\propto \sum_{\Gpre \in L(\sigma,G)} \Phi(\sigma,\Gpre) (1-\delta)^w \delta^m
\end{align*}
The proportionality constant hidden by `$\propto$' depends on $G$ but not on $\sigma$. As $\Gpre$ is obtained from $G$ by replacing $2m$ opposite-spin non-edges by opposite-spin edges, we have for every $\Gpre \in L(\sigma,G)$,
$$\Phi(\sigma,\Gpre) = \Phi(\sigma,G)\left(\frac{b/n}{1-b/n}\right)^{2m}.$$
Thus none of the terms depend on the precise choice of $\Gpre$:
\begin{align*}
\prob{\sigma \given G} &\propto \left|L(\sigma,G)\right| \Phi(\sigma,G) \left( \frac{b/n}{1-b/n} \right)^{2m} (1-\delta)^w \delta^m \\
&\propto \left|L(\sigma,G)\right| \Phi(\sigma,G) (1-\delta)^w.
\end{align*}
where we have dropped the constants that only depend on $G$ (and not $\sigma$).

Now we compute $|L(\sigma,G)|$. Suppose there are $m/2 + \alpha$ isolated \marked nodes of positive spin and $m/2 - \alpha$ of negative spin. Suppose there are $g/2 + \beta$ \good nodes of positive spin, and $g/2 - \beta$ of negative spin. Then the number of possible $\Gpre$ is
$$ |L(\sigma,G)| = \binom{g/2 + \beta}{2}^{m/2 - \alpha} \binom{g/2 - \beta}{2}^{m/2 + \alpha}. $$

\noindent We can establish that this global $L$ factor only barely depends on the spins of $A$:
\begin{lemma}\label{lemma:binom-ratio}
For a.a.e.\ $G, \sigma_B, \sigma_C$, it holds for all $\sigma_A$, $\sigma_A'$ that
$$ \frac{|L(\sigma_A,\sigma_B,\sigma_C,G)|}{|L(\sigma_A',\sigma_B,\sigma_C,G)|} = 1 + o(1). $$
\end{lemma}
\noindent The proof proceeds via concentration of measure and Taylor expansion, and is deferred to Appendix~\ref{sec:binom-ratio-proof}.

Paraphrasing this lemma, a.a.s.\ over $(\sigma_B,G)$, there exists a `good' subset of a.a.e.\ $\sigma_C$ such that $|L(\sigma,G)|$ is independent of $\sigma_A$ up to a $1 + o(1)$ factor. Let us also require of `good' $\sigma_C$ that the census (sum of spins) is $O(\sqrt{n} \log n)$; as $C$ consists of all but $O(n^{1/8})$ nodes of $G$, it is equivalent to ask for the same concentration over all spins of $G$, and this occurs a.a.s.\ by Hoeffding. Let $\Omega$ be this set of `good' $\sigma_C$ values. Since $\sigma_C \in \Omega$ a.a.s.\ we have for any set $D = D(G)$ (which will be taken as either $A$ or $A \cup B$ below),
\begin{equation}
\label{eq:Omega-fact}
\prob{\sigma_D,G} = (1+o(1)) \prob{\sigma_D,\sigma_C \in \Omega,G} \text{ a.a.s.\ over $\sigma,G$.}
\end{equation}
We include a rigorous proof of this statement in Appendix~\ref{sec:Omega-fact-proof}.

It will be useful to factor the product $\Phi$ of classical pairwise interactions into subsets of these interactions as
$$ \Phi(\sigma,G) = Q_{AB, AB}(\sigma_{AB},G) Q_{BC,C}(\sigma_{BC},G) Q_{A,C}(\sigma_{AC},G). $$
Here, for instance, $Q_{BC,C}$ denotes the product of $\phi(\sigma_u,\sigma_v,G)$ over unordered pairs $u,v$ consisting of one vertex from $B \cup C$ and one from $C$.
The corresponding ``no long-range correlations'' proof in \cite{mns} established that for `good' values $\sigma_C \in \Omega$, $Q_{A,C}(\sigma_{AC},G) = (1 + o(1)) K(G)$ for a quantity $K(G)$ depending only on $|A|$ and $|C|$. Their proof of this fact holds verbatim in our setting: it only requires that $|A| = o(n)$ and that the number of $+1$ spins in the graph is distributed as $\Binom(n,\frac12)$.

Similarly, we can factor the $(1-\delta)^w$ term as
$$(1-\delta)^{w(\sigma,G)} = (1-\delta)^{w_A(\sigma_{AB},G)}(1-\delta)^{w_C(\sigma_{BC},G)}$$
where, for instance, $w_A$ counts the number of \marked nodes in $A$ with two opposite-spin neighbors. This factorization holds as there are no $A$--$C$ edges and no \marked nodes in $B$. We can now absorb these terms into the $Q$ terms: define
$$Q'_{AB, AB}(\sigma_{AB},G) = (1-\delta)^{w_A(\sigma_{AB},G)}Q_{AB, AB}(\sigma_{AB},G), $$
and similarly $Q'_{BC,C}$, and $\Phi' = Q'_{AB,AB}Q'_{BC,C}Q_{A,C}$, so that for $\sigma_C \in \Omega$,
\begin{align*}
\prob{\sigma \given G} &\propto \Phi'(\sigma,G) \left|L(\sigma,G)\right| \\
&\propto (1+o(1)) \Phi'(\sigma,G) K(G) Q'_{AB,AB}(\sigma_{AB},G) Q'_{BC,C}(\sigma_{BC},G).
\end{align*}

At this point we roughly adapt the proof of conditional independence from factorization in a Markov random field, following the corresponding proof in \cite{mns}. We compute that for a.a.e.\ $(\sigma,G)$,
\begin{align*}
\prob{ \sigma_A \given \sigma_B, G }
&= \frac{\prob{\sigma_{AB}, G}}{\prob{\sigma_B, G}} \\
&= (1+o(1)) \frac{\prob{\sigma_{AB}, \sigma_C \in \Omega \given G}}{\prob{\sigma_B, \sigma_C \in \Omega \given G}} \qquad\text{using (\ref{eq:Omega-fact}),}\\
&= (1 + o(1)) \frac{\sum_{\sigma_C' \in \Omega} \Phi'(\sigma_{AB},\sigma'_C,G) |L(\sigma_{AB},\sigma'_C,G)|}{\sum_{\sigma_A', \sigma'_C \in \Omega} \Phi'(\sigma'_{AC},\sigma_B,G) |L(\sigma'_{AC},\sigma_B,G)|} \\
&= (1 + o(1)) \frac{ Q'_{AB, AB}(\sigma_{AB},G) \sum_{\sigma'_C \in \Omega} Q'_{BC,C}(\sigma_B,\sigma'_C,G) |L(\sigma_{AB},\sigma'_C,G)|}{ \sum_{\sigma_A'} Q'_{AB,AB}(\sigma'_A,\sigma_B,G) \sum_{\sigma'_C \in \Omega} Q'_{BC,C}(\sigma_B,\sigma'_C,G) |L(\sigma'_{AC},\sigma_B,G)| }, \\
&= (1 + o(1)) \frac{Q'_{AB,AB}(\sigma_{AB},G)}{\sum_{\sigma_A'} Q'_{AB,AB}(\sigma'_A,\sigma_B,G)} \qquad \text{using Lemma~\ref{lemma:binom-ratio}.} \\
\intertext{Multiplying the top and bottom by $K(G) Q'_{BC,C}(\sigma_{BC},G) |L(\sigma,G)|$:}
&= (1 + o(1)) \frac{ K(G) Q'_{AB,AB}(\sigma_{AB},G) Q'_{BC,C}(\sigma_{BC},G) |L(\sigma,G)|}{ \sum_{\sigma_A'} K(G)Q'_{AB,AB}(\sigma'_A,\sigma_B,G) Q'_{BC,C}(\sigma_{BC},G) |L(\sigma,G)|} \\
&= (1 + o(1)) \frac{ \Phi'(\sigma,G) |L(\sigma,G)| }{ \sum_{\sigma_A'} \Phi'(\sigma'_A,\sigma_{BC},G) |L(\sigma'_A,\sigma_{BC},G)|} \qquad \text{using Lemma~\ref{lemma:binom-ratio} again,} \\
&= (1 + o(1)) \prob{ \sigma_A \given \sigma_{BC}, G},
\end{align*}
as desired.
\end{proof}

\section{Random vs.\ Semirandom Separation in the Block Model}
\label{sec:hardness}

We can now assemble all of the pieces to prove our main result. We first prove that it is impossible to estimate the relative spin of any fixed pair of nodes, in a strictly larger parameter range than for the random model. The impossibility of partial recovery will then easily follow.

\begin{proposition}
\label{prop:rel-spin-hard}
For all $k > 1$, there exists $0 < \eps < \frac{1}{2}$ such that $k (1-2\eps)^2 > 1$, yet given a graph $G$ (including markings \good/\marked) from Distribution~\ref{dist:graph} with parameters $a = 2k(1-\eps)$, $b = 2k\eps$, we have that for any fixed vertices $u,v$,
$$ \prob{\sigma_u = +1 \given \sigma_v = +1, G} \to \frac{1}{2} $$
for a.a.e. $G$ as $n \to \infty$.
\end{proposition}
\noindent By `fixed' vertices $u,v$ we mean that the vertices are fixed before $\sigma$ and $G$ are chosen; so by symmetry it doesn't matter which $u,v$ pair we fix.
\begin{proof}
Given $k$, choose $\eps$ as in Proposition~\ref{prop:treehard} (tree separation). With probability $1 - o(1)$, the tree coupling of Proposition~\ref{prop:coupling} centered at vertex $u$ succeeds. In this case, the neighborhood $U$ of $u$ coupled to the tree is of size $O(n^{1/8})$, and so with probability $1-o(1)$, $v$ lies outside this neighborhood. Let $B$ be the boundary vertices of $U$, i.e.\ those vertices in $U$ corresponding to the leaves of the tree.
Let $A$ be the interior $U \setminus B$, and let $C$ be the complement of $U$ in $G$.

By the law of total variance,
\begin{equation}\label{eq:totvar} \Var( \sigma_u \given \sigma_v, G ) \geq \EE_{\sigma_{BC} \given \sigma_v, G}[\Var( \sigma_u \given \sigma_v, \sigma_{BC}, G )]. \end{equation}

\noindent With further probability $1 - o(1)$, Lemma~\ref{lemma:no-lrc} (no long-range correlations) succeeds, and we have
$$ \prob{\sigma_A \given \sigma_{BC}, G} = (1 + o(1)) \prob{ \sigma_A \given \sigma_B, G }, $$
so that
$$ \Var( \sigma_u \given \sigma_{BC}, G ) = (1 + o(1)) \Var( \sigma_u \given \sigma_B, G ). $$
But it follows from the coupling in Proposition~\ref{prop:coupling} (which includes spins and markings) that
$$ \Var( \sigma_u \given \sigma_B, G ) = (1 + o(1)) \Var( \sigma_\rho \given \sigma_{\text{leaves}}, T ). $$
Since $R \to \infty$ as $n \to \infty$, non-reconstruction in the tree model (Proposition~\ref{prop:treehard}) implies that this latter variance converges to $1$: the variance of a $\{\pm1\}$-valued random variable with expectation $o(1)$ is $1 - o(1)$.

So we know that, with probability $1 - o(1)$, a $1 - o(1)$ proportion of $\sigma_{BC}$ contribute a value $1 - o(1)$ to the expectation in (\ref{eq:totvar}). The remaining $o(1)$ proportion of $\sigma_{BC}$ contribute a value bounded in magnitude by $1$, as the variance of a $\{\pm 1\}$-valued variable must be. It follows that, a.a.s., 
$$ \Var( \sigma_u \given \sigma_v,G ) \geq \EE_{\sigma_{BC}}[\Var( \sigma_u \given \sigma_{BC},G )] = 1 - o(1), $$
and the only way that this is possible in the $\{\pm 1\}$-valued distribution $\sigma_u \given \sigma_v,G$ is if
$$ \prob{\sigma_u = +1 \given \sigma_v = +1, G} = \frac12 + o(1), $$
as desired.
\end{proof}

It will follow immediately from Proposition~\ref{prop:rel-spin-hard} that it is hard to find a partition that is correlated (better than random guessing) with the true one. For if it was possible to find such a partition then one could also guess the relative spins of $u$ and $v$. We now complete the proof of our main theorem, restating it slightly more precisely.

\begin{theorem}[restatement of Theorem~\ref{thm:intro-main}]
For any $k > 1$, there exists $0 < \eps < \frac{1}{2}$ so that $k (1- 2\eps)^2 > 1$ and hence partial recovery in the stochastic block model $G(n, a/n, b/n)$ is possible, and yet against the monotone adversary given in Section~\ref{sec:models}, for any $\eta > \frac12$, no estimator of the spins achieves $\eta$-partial recovery with probability greater than $o(1)$ as $n \to \infty$.
\end{theorem}
\begin{proof}
Let $\hat\sigma$ be some assignment of spins to vertices that achieves $\eta$-partial recovery: $\hat\sigma$ agrees with the true spins on at least $\eta n$ vertices, possibly after a global spin flip. Consider the relative spins $\hat\sigma_u \hat\sigma_v$, where $u$ and $v$ are distinct vertices; these match the true relative spins on at least
$$ \eta n (\eta n - 1) + (1-\eta) n ((1-\eta) n - 1) = (1 + o(1)) (\eta^2 + (1-\eta)^2) n^2 $$
ordered pairs of distinct vertices. If we choose two distinct vertices at random, the chance of correctly estimating their relative spin from $\hat\sigma$ is at least
\begin{equation}\label{eq:eta}
(1 + o(1)) (\eta^2 + (1-\eta)^2).
\end{equation}

Suppose for a contradiction that some estimator achieves $\eta$-partial recovery, for some $\eta > \frac{1}{2}$, with probability $p$ not converging to $0$ as $n \to \infty$. When $\eta$-partial recovery succeeds, the process above recovers the relative spin of two random vertices with probability at least $(1+o(1))(\eta^2+(1-\eta)^2)$, and note $\eta^2+(1-\eta)^2 > \frac{1}{2}$. When partial recovery does not succeed, the process still recovers the relative spin of two random vertices with probability at least $(1+o(1))\frac{1}{2}$, as can be seen by plugging in $\eta = \frac{1}{2}$ to (\ref{eq:eta}). It follows that we can recover the relative spin of two random vertices with probability at least $(1+o(1))[p(\eta^2+(1-\eta)^2) + (1-p)\frac{1}{2}]$ which remains bounded above $\frac{1}{2}$ as $n \to \infty$, contradicting Proposition~\ref{prop:rel-spin-hard}. (Proposition~\ref{prop:rel-spin-hard} assumes the markings on $G$ are known, but clearly the problem is only harder when they are unknown, since an estimator can choose to ignore them.)
\end{proof}

\section{Robustness of SDPs for Partial Recovery}
\label{sec:upper-bound}

We now turn to giving algorithms for partial recovery in the semirandom setting. In the existing literature on \emph{exact} recovery, it has been observed that algorithms obtained through semidefinite programming extend almost automatically to semirandom models \cite{fk,al,pw}, and moreover many of these results match the information-theoretic thresholds for exact recovery. In contrast, semidefinite programs for partial recovery, such as those of Gu\'edon and Vershynin \cite{GV}, come within a constant of the threshold but have been unable to close this gap.

In fact Gu\'edon and Vershynin \cite{GV} developed a general framework for using SDPs to solve problems on sparse graphs, including partial recovery for the stochastic block model. We define a notion of semirandom model for any problem in this framework, generalizing the semirandom model for community detection. We show that any analysis that follows their framework carries over automatically to this semirandom model. In particular, after minor modifications, the SDP analyzed in \cite{GV} achieves partial recovery in the semirandom block model, up to within a constant factor of the classical threshold.

Semirandom vs.\ random gaps offer an explanation for why it seems hard to find semidefinite programs for partial recovery that reach the information-theoretic threshold: the analysis often extends equally well to the semirandom model, where we know the threshold is strictly harder!

\subsection{The Gu\'edon--Vershynin Framework and Partial Recovery}
In the framework of \cite{GV}, there are several key terms. The first is a matrix $B$ which is usually computed in a simple manner from the adjacency matrix of the observed graph. In our case, we compute $B = A - \lambda J$, where $J$ is the all-ones matrix, $\lambda = \frac{a+b}{2n}$ is a regularization constant, and $A$ is the observed adjacency matrix (albeit with the non-standard convention that there are ones along the diagonal instead of zeros). The goal is to show that $B$ --- which is a random matrix --- is near some fixed reference matrix $R$. In our case this will be $\frac{a-b}{2n} \sigma \sigma^\top$, which is nearly equal to $\EE[B]$ except on the diagonal. Then, one solves an SDP to maximize $\langle B,Z \rangle$ over $Z \in \Omega$, where $\Omega$ is the space of symmetric PSD $n \times n$ matrices $Z$ satisfying $\diag(Z) \le I$ (i.e.\ the diagonal entries of $Z$ are at most $1$).\footnote{The framework also allows $\Omega$ to have additional constraints, but we will not need this.} The goal of the analysis is to show that the SDP outputs some $\hat Z$ that is close to a ``ground truth'' $Z_R$, which in our case is the $\{\pm 1\}$-valued matrix $\sigma \sigma^\top$ of true relative spins. The following result outlines the steps of the analysis.
\begin{proposition}
\label{prop:three}
Suppose we have $(B,R,Z_R)$ such that the following three conditions hold for some value $\alpha$, some function $F(\beta)$, and some matrix norm $\|\cdot \|$:
\begin{enumerate}[(1)]
\item $Z_R$ is a maximizer of the reference objective $\langle R, - \rangle$ over $\Omega$ \\ (the reference SDP recovers the truth),
\item $\| B-R \|_{\infty \to 1} \leq \alpha$ \\ (the observed objective is close to the reference one in cut norm),
\item if $Z \in \Omega$ and $\langle R, Z_R - Z \rangle \leq \beta$, then $\| Z_R - Z \|^2 \leq F(\beta)$ \\
(good solutions to the reference SDP are close to the ground truth).
\end{enumerate}
Then $\| Z_R - \hat Z \|^2 \leq F(2 K_G \alpha)$ where $\hat Z$ is any maximizer of the empirical objective $\langle B, - \rangle$ over $\Omega$, and $K_G$ is the Grothendieck constant.
\end{proposition}
\noindent Here, the cut norm (or $\infty$-to-$1$ norm) of a matrix is defined as
$$\|M\|_{\infty \to 1} \defeq \max_{\|x\|_\infty = 1} \|Mx\|_1 = \max_{x \in \{\pm 1\}^n} \|Mx\|_1.$$
The proof of Proposition~\ref{prop:three} follows from Lemma~3.3 in \cite{GV}, and uses Grothendieck's inequality.

The partial recovery results of \cite{GV} proceed by verifying the three conditions of Proposition~\ref{prop:three} for a particular choice of parameters:
\begin{proposition}
\label{prop:check-three}
Assume $a > 20$. Let $B = A - \lambda J$ with $\lambda = \frac{a+b}{2n}$, $R = \frac{a-b}{2n}\sigma \sigma^\top$, and $Z_R = \sigma \sigma^\top$. Conditions (1--3) of Proposition~\ref{prop:three} hold a.a.s.\ with $\alpha = O(n \sqrt{a+b}), F(\beta) = \beta \cdot O(n/(a-b))$, and $\|\cdot\|$ as the Frobenius norm $\|\cdot\|_2$.
\end{proposition}
\noindent Concretely, this means that we solve the following SDP:
\begin{sdp}\label{sdp:gv}
Maximize $\langle A - \lambda J, Z \rangle$ subject to $Z \succeq 0$ and $\diag(Z) \leq I$, where $\lambda = \frac{a+b}{2n}$.
\end{sdp}
\noindent Note that the regularization constant $\lambda$ is necessary because if we simply take $B = A$ then condition 1 of Proposition~\ref{prop:three} fails. In \cite{GV} they estimate $\lambda$ from the empirical average degree, but their arguments also apply to the case where we deterministically take $\lambda = \frac{a+b}{2n}$. (This requires knowledge of the parameters $a,b$ but we will address this issue later.)

In \cite{GV}, condition 1 is shown in Lemma~5.1, and conditions 2 and 3 are implicit in the proof of Lemma~5.2. These require a technical condition: $\max\{a(1-a/n),b(1-b/n)\} \ge 20$, which for large $n$ simply amounts to $a > 20$. By Propositions \ref{prop:three} and \ref{prop:check-three}, we now attain the result $\|Z_R - \hat Z\|_2^2 \le O(n^2 \sqrt{a+b}/(a-b))$ a.a.s. It is also shown in \cite{GV} how to translate this to a precise partial recovery result:
\begin{proposition}\label{prop:rounding}
Let $Z_R = \sigma \sigma^\top$. For any $\frac12 < \eta \le 1$, $\eta$-partial recovery succeeds a.a.s.\ in the stochastic block model, provided that $\|Z_R - \hat Z\|_2^2 \le (1-\eta)n^2$, by taking the signs in the top eigenvector of $\hat Z$.
\end{proposition}
\begin{corollary}\label{cor:random-recovery}
There exists a constant $C$ such that $\eta$-partial recovery succeeds a.a.s.\ in the stochastic block model as $n \to \infty$ whenever $a > 20$ and $(a-b)^2 \ge C (1-\eta)^{-2} (a+b)$.
\end{corollary}
\noindent The constant $C$ is quite large: \cite{GV} sets $C = 10^4$, although no attempt is made to optimize this constant. Nevertheless this result is only off by a constant from the threshold, which is $(a-b)^2 > 2(a+b)$.
Our random vs.\ semirandom separation becomes small as $k$ grows, so it remains plausible that semidefinite programs can achieve partial recovery when $k (1-2\eps)^2 > 1 + o(1)$ as $k \to \infty$. In fact, SDPs are known to distinguish random block model graphs from Erd\H{o}s--R\'enyi graphs in such a range \cite{ms}, and it would be interesting to determine whether this carries through to partial recovery in the semirandom model.

\subsection{SDPs and Semirandom Models}
We now turn to a semirandom view of the general Gu\'edon--Vershynin framework. In the semirandom block model, a monotone adversary is allowed to make changes aligned with the ground truth $Z_R$; more formally, it can add a matrix $S$ to the observed adjacency $A$, where $S$ is symmetric and has $1$'s in some same-spin entries where $A$ is $0$, and $-1$'s in some opposite-spin entries where $A$ is $1$. It is easily verified that $Z_R$ maximizes $\langle S, - \rangle$ over $\Omega$: every matrix in $\Omega$ has entries in $[-1,1]$, and $Z_R$ has $\pm 1$ entries that match the sign pattern of $S$.

Following this observation, we propose a precise definition of semirandom models for Gu\'edon--Vershynin problems:
\begin{definition}
A matrix $S$ is a {\bf monotone change} if $Z_R$ is a maximizer for $\langle S,- \rangle$ over $\Omega$. In the {\bf semirandom model}, a preliminary objective $B$ is generated from the random model, and then an adversary modifies the objective by any monotone change $S$, yielding an observed matrix $B+S$.
\end{definition}
\noindent Note that the set of such monotone changes $S$ forms a cone, which matches the intuition that semirandom models can make unbounded changes, but only in certain directions aligned with the ground truth.

Any analysis following the framework of Proposition~\ref{prop:three} generalizes in a formal sense to the semirandom model:
\begin{proposition}
\label{prop:S}
Suppose that $(B,R,Z_R)$ satisfy conditions (1--3) of Proposition~\ref{prop:three} for some $\alpha,F(\beta),\|\cdot\|$. Then it remains the case that $\| Z_R - \hat Z_S \|^2 \leq F(2 K_G \alpha)$, where $\hat Z_S$ is any maximizer of $\langle B + S, - \rangle$ over $\Omega$, and $S$ is any monotone change (in the sense above), which is allowed to depend adversarially on $B$.
\end{proposition}
\begin{proof}
It suffices to establish conditions (1--3) of Proposition~\ref{prop:three} with $R + S$ playing the role of $R$ and $B + S$ playing the role of $B$:
\begin{enumerate}[(1)]
\item As $Z_R$ maximizes both $\langle R, - \rangle$ and $\langle S, - \rangle$ over $\Omega$, it certainly maximizes their sum $\langle R + S, - \rangle$.
\item We have $\| (B + S) - (R + S) \|_{\infty \to 1} = \| B - R \|_{\infty \to 1}$, which by hypothesis is at most $\alpha$ with high probability.
\item Suppose that $Z \in \Omega$ and $\langle R + S, Z_R - Z \rangle \leq \beta$. Then
$$ \langle R, Z_R - Z \rangle \leq \beta - \langle S, Z_R - Z \rangle \leq \beta, $$
since $Z_R$ is a maximizer of $\langle S,- \rangle$ over $\Omega$. It follows from the original condition (3) that $\| Z_R - Z \|^2 \leq F(\beta)$ in the appropriate norm.
\end{enumerate}
\end{proof}

We have shown that any {\em analysis} of an SDP following the Gu\'edon--Vershynin framework transfers to the semirandom model. A recovery guarantee for the same SDP by different means would not automatically yield such a guarantee in the semirandom model. This differs from exact recovery problems, where semirandom guarantees can take the success of the SDP as a black box, and certify that any successful instance remains successful under monotone changes \cite{fk}.

\begin{corollary}
Let $A$ be sampled from the semirandom block model with parameters $a,b,n$. Suppose that we know these parameters, so that we can run SDP~\ref{sdp:gv} with the deterministic choice $\lambda = \frac{a+b}{2n}$. If $a > 20$ and $(a-b)^2 \geq 10^4 (1-\eta)^{-2} (a+b)$ then this SDP achieves $\eta$-partial recovery a.a.s.\ (after taking the sign of the top eigenvector of the SDP output $\hat Z$).
\end{corollary}

This dependence on parameter knowledge is slightly unwieldy; on the other hand, estimating the correct regularization parameter $\lambda$ from $A + S$ is not immediately easy, since model-specific estimators are precisely the sort of techniques that semirandom models aim to penalize. Instead, we show that deterministically taking $\log n / n$ in place of $\lambda$ avoids this dependence. Indeed, the same SDP appears in the exact recovery literature \cite{abh} with deterministic regularization parameter $\frac12$, which is comparatively very large, but comes at the cost of requiring precisely balanced communities. Our approach here is a middle ground that allows for natural $\sqrt{n}$-scale deviations in the balance of the two communities. There is nothing special about the value $\log n/n$ and it can in fact be taken to be anything of order strictly between $1/n$ and $1/\sqrt n$.

\begin{proposition}\label{prop:new-sdp}
The Gu\'edon--Vershynin SDP (SDP~\ref{sdp:gv}) with regularization parameter $\lambda' = \log n / n$ achieves $\eta$-partial recovery against the semi-random model a.a.s.\ so long as $a > 20$ and $(a-b)^2 \geq 10^4 (1-\eta)^{-2} (a+b)$.
\end{proposition}
\begin{proof}
In light of Proposition~\ref{prop:S}, it is sufficient to show that this SDP works in the random model. We check the three conditions for $(B',R',Z_R)$ where $B' = A - \lambda' J = B + (\lambda - \lambda')J$, $R' = \frac{a-b}{2n} \sigma \sigma^\top + (\lambda - \lambda')J = R + (\lambda - \lambda')J$, and $Z_R = \sigma \sigma^\top$ where $\lambda = \frac{a+b}{2n}$ and $\lambda' = \frac{\log n}{n}$. We will use the fact that the three conditions are satisfied for $(B,R,Z_R)$, and we will achieve nearly the same parameters $\alpha,F(\beta)$ as in that case.

\begin{enumerate}
\item[(1)] We certify that $Z_R$ maximizes $\langle R', - \rangle$ by SDP duality. The dual SDP reads
$$ \text{minimize }\sum_{v} \gamma_v \quad  \text{subject to }\Lambda \defeq \diag(\gamma) - R' \succeq 0, $$
where $v$ runs over all vertices in the graph. 
To certify that $Z_R$ is optimal, it suffices by complementary slackness to find $\gamma$ such that $\Lambda \succeq 0$ and $\langle \Lambda , Z_R \rangle = 0$. Since $\Lambda$ and $Z_R$ are PSD, the second condition is equivalent to $\Lambda Z_R = 0$. As the columns of $Z_R$ are spanned by $\sigma$, this second condition reads $\Lambda \sigma = 0$; expanding,
$$ \gamma_v = \frac{a-b}{2} + \left(\frac{a+b}{n} - \frac{\log n}{n}\right)\left(|C_v| - |\bar C_v|\right) \quad \forall v $$
where $C_v$ is the set of nodes with the same spin as node $v$, and $\bar C_v$ is the complement of $C_v$. By Hoeffding we have $\left||C_v|-|\bar C_v|\right| \le \sqrt{n}\log n$ a.a.s.\ and so
$$ \gamma_v = \frac{a-b}{2} - O(\log^2 n/\sqrt{n}) = \frac{a-b}{2} - o(1). $$
In particular, $\gamma_v \ge 0$ for all $v$ a.a.s. This choice of $\gamma$ guarantees $\Lambda \sigma = 0$ so it remains to show that $\Lambda \succeq 0$. Since $\Lambda \sigma = 0$, it suffices to show that
$v^\top \Lambda v \geq 0$ for all $v \perp \sigma$. But note that
$$ \Lambda = \left( \frac{\log n}{n} - \frac{a+b}{2n} \right) J - \frac{a-b}{2n} \sigma \sigma^\top + \diag(\gamma), $$
$$ v^\top \Lambda v = v^\top \left[ \left( \frac{\log n}{n} - \frac{a+b}{2n} \right) J + \diag(\gamma) \right] v \geq 0, $$
since $\gamma \geq 0$ a.a.s.\ so that both $J$ and $\diag(\gamma)$ are PSD.
This establishes optimality of $Z_R$.

\item[(2)] We have $B' - R' = B - R$ and so this step follows from the corresponding step for $(B,R,Z_R)$.

\item[(3)] Suppose that $Z \in \Omega$ satisfies $\langle R', Z_R - Z \rangle \leq \beta$. Letting $\nu = \lambda' - \lambda$ we have $0 \le \nu \le \frac{\log n}{n}$ and $R' = R - \nu J$. Notice that $\langle Z, J \rangle \geq 0$ since both matrices are PSD, while $\langle Z_R, J \rangle = \frac{1}{2}(|C_+| - |\bar C_-|)^2 \leq n \log^2 n$ a.a.s.\ where $C_+,C_-$ are the community sizes. Thus,
\begin{align*}
\langle R, Z_R - Z \rangle
&= \langle R', Z_R - Z \rangle + \nu \langle J, Z_R - Z \rangle \\
&\le \beta + \nu \cdot n \log^2 n - \nu \langle J, Z \rangle \\
&\le \beta + \log^3 n.
\end{align*}
Applying condition 3 for $(B,R,Z_R)$, we obtain $\|Z_R - Z\|_2^2 \leq (\beta + \log^3 n) \cdot O(n/(a-b))$.
\end{enumerate}

By Proposition~\ref{prop:rounding}, we obtain partial recovery by rounding the leading eigenvector, under the same conditions as in Corollary~\ref{cor:random-recovery}, and with the same constant $C$.
\end{proof}

\noindent This analysis completes the proof of Theorem~\ref{thm:intro-robust-sdp} from the introduction.

Independently, \cite{mmv-robust} proved a result on SDP robustness that matches ours. Additionally, they consider robustness to a different adversary that can only add or remove $o(n)$ edges but is not required to be monotone. Our methods also extend to this case (and match their results) because $o(n)$ changes can only change the cut norm by $o(n)$.

\section{Robustness of Recursive Majority in the Broadcast Tree Model}\label{sec:recmaj}

Here we establish that recursive majority achieves reconstruction in the semirandom broadcast tree model, even with respect to the strong adversary that has total control over entire subtrees. As one would imagine, in recursive majority the root spin is estimated as the majority vote of its children's estimated spins, which in turn are estimated as the majority vote of their children, and so on down to the known leaf spins. In the random model, this algorithm falls short of the Kesten--Stigum bound by a factor of $\sqrt{2/\pi}$ as $k \to \infty$ \cite{mossel-recursive}. Nevertheless, it remains the algorithm of choice in practice for tree reconstruction problems, often by the name of `parsimony' \cite{mossel-survey}. It seems that the advantages of recursive majority over majority only become clear when studying these algorithms through semirandom models!

To keep things simple, we will first consider trees in which each non-leaf node has exactly $k$ children, matching the setting of \cite{mossel-recursive} and much of the tree reconstruction literature. At the end, we show how to adapt our proof in a straightforward manner to the case where the number of children of each non-leaf node is a Poisson random variable.

\begin{proposition}
\label{prop:eps-star}
Let $M_k(p)$ be the probability of a majority `yes' vote of $k$ voters each voting `yes' with probability $p$:
$$ M_k(p) \defeq \problr{\Binom(k,p) > \frac{k}{2}} + \frac12 \problr{\Binom(k,p) = \frac{k}{2}}. $$
Let $(T,\sigma)$ be drawn from the strong semirandom model on the regular tree with braching factor $k$, height $R$, and noise $\eps$. Let $\eps^*_k$ be defined by
$$ \frac{1}{1-\eps^*_k} = \max_{q \in (0,1]} \frac{M_k(q)}{q}, $$
and let $q^*_k$ be the maximizer. So long as $\eps \leq \eps^*_k$, the recursive majority algorithm correctly recovers the root spin from the topology and leaf spins, with probability at least $p^*_k = q^*_k / (1-\eps^*_k) > \frac12$. Conversely, if $\eps > \eps^*_k$, there exists an adversary forcing recursive majority to fail with probability $1 - o(1)$ as $R \to \infty$.
\end{proposition}
\begin{proof}
Let $(T,\sigma)$ be drawn from the strong semirandom model, starting with a $k$-regular tree of height $R$. Let $p_R$ denote the success probability on such a tree. As a base case, we have $p_0 = 1$. For each child of the root, the spin matches that of the root with probability $(1-\eps)$, in which case the child will be correctly labeled by recursive majority with probability $p_{R-1}$, by induction. The other children are under adversary control. Then the number of recovered child labels agreeing with the root label stochastically dominates $\Binom(k,p_{R-1}(1-\eps))$, for any choice of adversary. The recovered root label is the majority vote of the recovered child labels, so we obtain that
$$ p_R \geq M_k(p_{R-1} (1-\eps)). $$
Moreover this inequality is exactly tight for a certain adversary, which replaces each replaceable subtree by a path down to a single leaf of spin opposite to the root.

We are interested in the limiting recovery probability $p_\infty = \lim_{R \to \infty} p_R$; this is the greatest fixed point in $[0,1]$ of the function $p \mapsto M_k(p(1-\eps))$. Equivalently, $q_\infty = p_\infty (1-\eps)$ is the greatest solution in $[0,1-\eps]$ of the fixed-point equation $q/(1-\eps) = M_k(q)$. Geometrically, this is the greatest intersection point of the graph of $M_k$ with the line of slope $1/(1-\eps)$ through the origin.

Before proceeding, it is worth noting the geometry of the graph of $M_k$. From the definition, $M_k$ is monotone on $[0,1]$, with $M_k(0) = 0$, $M_k(\frac12) = \frac12$, and $M_k(1) = 1$. By expanding the probability mass function of the binomial distribution, $M_k$ is a degree $k$ polynomial, and thus smooth. Less obviously, $M_k$ is strictly convex on $[0,\frac12]$ and strictly concave on $[\frac12,1]$ \cite{mossel-recursive}.

The line of slope $2$ ($\eps = \frac12$) does not intersect the graph of $M_k$ nontrivially, since $M_k(q) \leq \frac12$ on $[0,\frac12]$ and $M_k(q) \leq 1$ on $[\frac12,1]$. The line of slope $1$ ($\eps = 0$) certainly intersects the graph of $M_k$ at $\frac12$ and $1$. Hence there exists some maximal $0 < \eps^*_k < \frac12$, at which the line of slope $1/(1-\eps^*_k)$ intersects the graph of $M_k$ tangentially at some $q^*_k$. Equivalently, we can characterize this slope as the maximum slope of the line defined by the origin and any point on the graph of $M_k$:
$$ \frac{1}{1-\eps^*_k} = \max_{q \in (0,1]} \frac{M_k(q)}{q}. $$
By the concavity and convexity properties of $M_k$, the graph of $M_k$ lies below the line of slope $1$ on $(0,\frac12)$, and above on $(\frac12,1)$. It follows that $q^*_k > \frac12$.

As we sweep $q$ from $q^*_k$ to $1$, the slope $R(q)/q$ passes continuously from $1/(1-\eps^*_k)$ to $1$, granting an intersection $q \geq q^*_k$ satisfying $R(q)/q = 1/(1-\eps)$ for every $\eps \in [0,\eps^*_k]$, by the intermediate value theorem. Thus, whenever $0 \leq \eps \leq \eps^*_k$, recursive majority achieves recovery with probability
$$ p_\infty = \frac{q_\infty}{1-\eps} \geq q_\infty \geq q^*_k > \frac12. $$
Whenever $\eps > \eps^*_k$, the two curves do not intersect on $[0,1]$ except at $0$, so the limiting success probability $p_\infty$ is $0$.
\end{proof}

It is interesting that, as the noise $\eps$ varies, we see the success probability jump discontinuously from $p^*_k > \frac12$ to zero. This contrasts with the behavior of recursive majority in the ordinary broadcast tree model, in which the error probability transitions continuously to $\frac12$ at a threshold \cite{mossel-recursive}.

\vspace{1em}\fbox{ \parbox{0.90\textwidth}{\begin{center}
\includegraphics[width=0.45\textwidth]{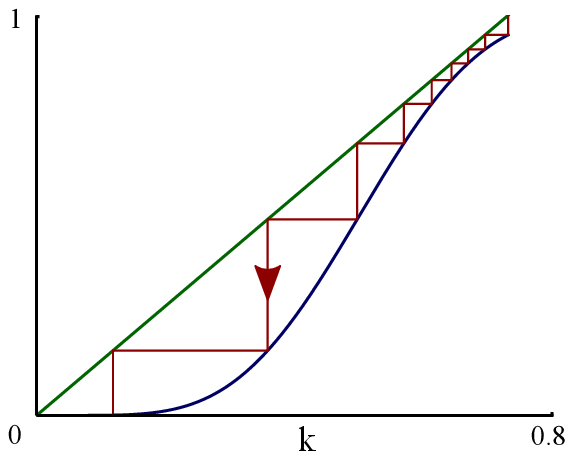}

The curves $R(q)$ and $q/(1-\eps)$ for $k=11$, $\eps = 0.25$. A cobweb plot shows the effect of iterating $q \mapsto R(q)/(1-\eps)$; here there is no nontrivial fixed point, and recursive majority fails. A small decrease in $\eps$ will cause a fixed point to occur near $q=0.683$.
\end{center}}}\vspace{1em}

\vspace{1em}\noindent\fbox{ \parbox{0.97\textwidth}{\begin{center}
\begin{minipage}[b]{0.45\textwidth}
\includegraphics[width=\textwidth]{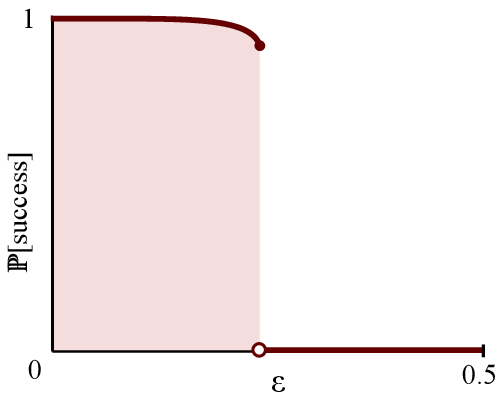}
\end{minipage}
\begin{minipage}[b]{0.45\textwidth}
\includegraphics[width=\textwidth]{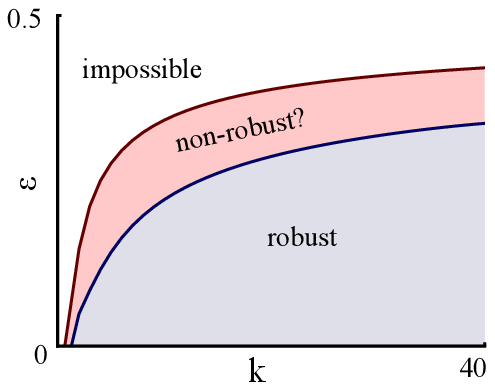}
\end{minipage}\vspace{0.3em}

Left: success probability in the $11$-regular tree. Right: the Kesten--Stigum threshold and the semirandom recursive majority threshold in $\Pois(k)$-birth trees.
\end{center}}}\vspace{1em}

For any fixed $k$, computing the critical value $\eps^*_k$ amounts to maximizing the polynomial $\frac{M_k(q)}{q}$, or equivalently finding the unique root of its derivative in $(0,1]$. For example it is easily computed that $\eps^*_3 = \frac19$. However, we would like some asymptotic understanding of the values $\eps^*_k$ as $k \to \infty$.

\begin{proposition}
The following asymptotic expression holds for $\eps^*_k$ as $k \to \infty$:
$$ \eps^*_k = \frac12 - \left(\frac12 + o(1) \right) \sqrt{\frac{\log k}{k}}. $$
\end{proposition}
\begin{proof}
By the maximality property defining $\eps^*_k$ and $q^*_k$, we must have
$$ 0 = \frac{\mathrm{d}}{\mathrm{d}q} \frac{M_k(q)}{q} \big|_{q=q^*_k} = (q^*_k)^{-2} (q^*_k \cdot M_k'(q^*_k) - M_k(q^*_k)), $$
from which it follows that
$ M_k'(q^*_k) = \frac{M_k(q^*_k)}{q^*_k} = \frac{1}{1-\eps^*_k}.$ Thus, as $\eps^*_k \in [0,\frac12]$, we know that
$$ 1 = \frac{1}{1 - 0} \leq M_k'(q^*_k) = \frac{M_k(q^*_k)}{q^*_k} \leq \frac{1}{1 - \frac12} = 2. $$

In analyzing recursive majority in the random model, Mossel \cite{mossel-recursive} computed this derivative using Russo's formula. Let us assume for the moment that $k$ is odd; the even case is similar. Then
$$ M_k'(q) = (q(1-q))^{(k-1)/2} k \binom{k-1}{(k-1)/2}. $$
If we evaluate this at $q = \frac12 + \frac12 \sqrt{\alpha \log k / k}$, we obtain:
\begin{align*}
M_k'(q) &= \left( \frac14 - \frac14\,\frac{\alpha \log k}{k} \right)^{(k-1)/2} k \binom{k-1}{(k-1)/2} \\
&= (1 + o(1)) \left(1 - \frac{\alpha \log k}{k} \right)^{(k-1)/2} \sqrt{\frac{k}{\pi}} \\
&= (1 + o(1)) \left( 1 - \frac{\frac12 \alpha \log k}{(k-1)/2} \right)^{(k-1)/2} \sqrt{\frac{k}{\pi}} \\
&= (1 + o(1)) \exp(- \frac12 \alpha \log k) \sqrt{\frac{k}{\pi}} \\
&= (1 + o(1)) \sqrt{\frac1\pi} k^{\frac12 - \frac12 \alpha}.
\end{align*}
In the second-to-last step, we have used the asymptotic identity $(1-x/k)^k = (1 + o(1)) \exp(-x)$ as $k \to \infty$, which still holds when $x$ depends on $k$ so long as $x = o(\sqrt{k})$. Now for $\alpha > 1$, the derivative $M_k'(q)$ tends to $0$ as $k \to \infty$, while for $\alpha < 1$, this tends to $\infty$ as $k \to \infty$. Thus, writing $q^*_k = \frac12 + \frac12 \sqrt{\alpha^*_k \log k / k}$, we must have $\alpha^*_k = 1 + o(1)$.
We immediately obtain a lower bound on $\eps^*_k$:
$$ \eps^*_k = 1 - \frac{q^*_k}{M_k(q^*_k)} \geq \frac{q^*_k}{1} = \frac12 + \left(\frac12 + o(1) \right) \sqrt{\frac{\log k}{k}}. $$
For a matching upper bound, apply Hoeffding's inequality for a binomial tail bound, to find
$$ M_k(q) \geq 1 - \exp(-\alpha \log k / 2) = 1 - k^{-\alpha/2}, $$
and use the maximality of $q^*_k$ to find
\begin{align*}
\eps^*_k &= 1 - \frac{q^*_k}{M_k(q^*_k)} \leq 1 - \frac{q}{M_k(q)} \\
&\leq 1 - \frac{\frac12 + \frac12 \sqrt{\alpha \log k / k}}{1 - k^{-\alpha/2}} \\
&= 1 - \left(\frac12 + \frac12 \sqrt{\frac{\alpha \log k}{k}} \right) \sum_{j=0}^\infty (k^{-\alpha/2})^j \\
&= \frac12 - \frac12 \sqrt{\frac{\alpha \log k}{k}} - O(k^{-\alpha/2}),
\end{align*}
so that $\frac12 - (\frac12 + \beta) \sqrt{\frac{\log k}{k}}$ is an asymptotic upper bound for every $\beta > 0$, which completes the proof.
\end{proof}
Note that $p_k^* = q_k^*/(1-\eps)$ is the probability of success at the threshold $\eps^*_k$, so this proof also provides some sense of the critical success probability:
$$ p_k^* = \frac{ \frac12 + (\frac12 + o(1) \sqrt{\log k / k}}{ 1 - \left( \frac12 - (\frac12 + o(1)) \sqrt{\log k / k} \right)} = 1 - o\left(\sqrt{\frac{\log k}{k}}\right). $$
So we observe a very strong threshold: as we vary $\eps$, there is a discrete jump from very likely success of recursive majority to almost-sure failure!

This result does not change if we pass back to the Poisson-birth tree used throughout the rest of this paper, rather than the $k$-regular tree. Here $R_{\Pois(k)}(q) = \EE_{\ell \sim \Pois(k)}[R_\ell(q)]$, which leads to a Poisson-averaged derivative, and the same threshold in $q$ for when the derivative passes from large to small. By Poissonization, we can also write
$$ R_{\Pois(k)}(q) = \problr{ \Pois(kq) > \Pois(k(1-q)) } + \frac12 \problr{ \Pois(kq) = \Pois(k(1-q)) }, $$
and the standard Chernoff bound for ``Poisson races'' yields, for $q = \frac12 + \nu$,
\begin{align*}
1 - R_{\Pois(k)}(q) &\leq \exp\left(-k\left(\sqrt{\frac12 + \nu} - \sqrt{\frac12 - \nu} \right)^2 \right) \\
&= \exp(-k(1-\sqrt{1-4\nu^2})) \\
&= \exp(-2 k (\nu^2 + O(\nu^4))),
\end{align*}
which is as strong as the Hoeffding bound used above.

The results above constitute a proof for Theorem~\ref{thm:intro-recmaj}: recursive majority achieves recovery against the strong semirandom model, so long as
$$ \eps < \frac12 - \left( \frac12 + o(1) \right) \sqrt{\frac{\log k}{k}}, $$
or, rearranging,
$$ \frac{k}{\log k}(1 - 2 \eps)^2 > 1 + o(1), $$
with asymptotics holding as $k \to \infty$.

It is interesting to see how the semirandom model can yield recovery results across an entire family of previously studied distributions. The asymmetric broadcast trees studied in \cite{borgs-asymmetric} may be described as Markov chains on the tree with transition matrix
$$ \mtx{1-\eps + \delta & \eps + \delta \\ \eps - \delta & 1-\eps - \delta}, $$
where the `asymmetry' $\delta$ may be positive or negative. The strong semirandom adversary can simulate these models: starting from the broadcast tree model with noise $\eps + |\delta|$, an adversary searches the tree top-down for either $+1$-to-$-1$ or $-1$-to-$+1$ transitions, depending on the sign of $\delta$, and flips the entire resulting subtree with probability $2|\delta|/(\eps + |\delta|)$, recursing into subtrees. Thus, the semirandom recovery results above guarantee for free that recursive majority is competitive against all of these models simultaneously, as are all algorithms robust to the strong semirandom model. By contrast, these models each admit a simple but brittle ``weighted majority'' reconstruction algorithm, where the weighting must be perfectly tuned to the specific model --- a tiny deviation in the parameters $(\eps,\delta)$ will cause those census algorithms to fail as $R \to \infty$.

\section{Conclusion}

In revisiting the stochastic block model from the perspective of semirandom models, we showed that there is a tension between establishing sharp thresholds and obtaining algorithms with natural robustness guarantees. There are many more classical problems in statistics and machine learning where semirandom models could offer a promising way to move beyond average-case analysis and explore issues related to robustness. In particular, belief propagation is one of the most far-reaching heuristics, but at high `temperature' it can result in algorithms such as majority that we have shown are not robust. At low `temperature' it leads to more robust algorithms, like recursive majority, that are connected to convex optimization. Can exploring `temperature' further lead to interesting, provable tradeoffs between robustness and statistical power that result in a richer understanding of how belief propagation performs in practice?

\section*{Acknowledgements}
The authors would like to thank Philippe Rigollet and the MIT learning theory group for helpful discussion, and Roxane Sayde for reading a draft of this document.

\bibliographystyle{alpha}
\bibliography{}

\begin{thebibliography}{MMV15b}

\bibitem[ABH15]{abh}
E.~Abbe, A.~S. Bandeira, and G.~Hall.
\newblock Exact recovery in the stochastic block model.
\newblock {\em IEEE Transactions on Information Theory}, 2015.

\bibitem[AL14]{al}
A.~A. {Amini} and E.~{Levina}.
\newblock On semidefinite relaxations for the block model.
\newblock {\em arXiv:1406.5647}, June 2014.

\bibitem[AS00]{alon-spencer}
N.~{Alon} and J.~H. {Spencer}.
\newblock {\em The probabilistic method}.
\newblock Wiley-Interscience series in discrete mathematics and optimization.
  Wiley, New York, Chichester, Weinheim, 2000.

\bibitem[AS15]{as}
E.~{Abbe} and C.~{Sandon}.
\newblock Community detection in general stochastic block models: fundamental
  limits and efficient recovery algorithms.
\newblock In {\em 56th Annual {IEEE} Symposium on Foundations of Computer
  Science ({FOCS})}, 2015.

\bibitem[BCLS87]{bui}
T.~N. Bui, S.~Chaudhuri, F.~T. Leighton, and M.~Sipser.
\newblock Graph bisection algorithms with good average case behavior.
\newblock {\em Combinatorica}, 7(2):171--191, 1987.

\bibitem[BCMR06]{borgs-asymmetric}
C.~{Borgs}, J.~{Chayes}, E.~{Mossel}, and S.~{Roch}.
\newblock The {Kesten}-{Stigum} {Reconstruction} {Bound} {Is} {Tight} for
  {Roughly} {Symmetric} {Binary} {Channels}.
\newblock In {\em 47th Annual {IEEE} Symposium on Foundations of Computer
  Science ({FOCS})}, pages 518--530, 2006.

\bibitem[{Bop}87]{boppana}
R.~B. {Boppana}.
\newblock Eigenvalues and graph bisection: An average-case analysis.
\newblock In {\em 28th Annual Symposium on Foundations of Computer Science
  ({FOCS})}, pages 280--285, 1987.

\bibitem[BS95]{bs}
A.~{Blum} and J.~{Spencer}.
\newblock Coloring random and semi-random k-colorable graphs.
\newblock {\em J. Algorithms}, 19(2):204--234, 1995.

\bibitem[DKMZ11]{decelle}
A.~Decelle, F.~Krzakala, C.~Moore, and L.~Zdeborov\'a.
\newblock Asymptotic analysis of the stochastic block model for modular
  networks and its algorithmic applications.
\newblock {\em Physical Review E}, 84(6):066106, December 2011.

\bibitem[EKPS00]{evans}
W.~{Evans}, C.~{Kenyon}, Y.~{Peres}, and L.~J. {Schulman}.
\newblock Broadcasting on trees and the {Ising} model.
\newblock {\em The Annals of Applied Probability}, 10(2):410--433, May 2000.

\bibitem[FK00]{fkr}
U.~{Feige} and R.~{Krauthgamer}.
\newblock Finding and certifying a large hidden clique in a semirandom graph.
\newblock {\em Random Struct. Algorithms}, 16(2):195--208, 2000.

\bibitem[FK01]{fk}
U.~{Feige} and J.~{Kilian}.
\newblock Heuristics for semirandom graph problems.
\newblock {\em Journal of Computing and System Sciences}, 63:639--671, 2001.

\bibitem[GV15]{GV}
O.~Gu\'edon and R.~Vershynin.
\newblock Community detection in sparse networks via {Grothendieck}'s
  inequality.
\newblock {\em Probability Theory and Related Fields}, 2015.

\bibitem[HLL83]{holland}
P.~W. Holland, K.~B. Laskey, and S.~Leinhardt.
\newblock Stochastic blockmodels: First steps.
\newblock {\em Social networks}, 5(2):109--137, 1983.

\bibitem[HWX15]{hwx}
B.~{Hajek}, Y.~{Wu}, and J.~{Xu}.
\newblock Achieving exact cluster recovery threshold via semidefinite
  programming.
\newblock In {\em Information Theory (ISIT), 2015 IEEE International Symposium
  on}, pages 1442--1446, June 2015.

\bibitem[JMR15]{jm-phase}
A.~{Javanmard}, A.~{Montanari}, and F.~{Ricci-Tersenghi}.
\newblock {Phase Transitions in Semidefinite Relaxations}.
\newblock {\em arXiv:1511.08769}, November 2015.

\bibitem[JS98]{jerrum}
M.~{Jerrum} and G.~B. {Sorkin}.
\newblock The metropolis algorithm for graph bisection.
\newblock {\em Discrete Applied Mathematics}, 82(1-3):155--175, 1998.

\bibitem[KMM11]{kmm}
A.~{Kolla}, K.~{Makarychev}, and Y.~{Makarychev}.
\newblock How to play unique games against a semi-random adversary: Study of
  semi-random models of unique games.
\newblock In {\em {IEEE} 52nd Annual Symposium on Foundations of Computer
  Science ({FOCS})}, pages 443--452, 2011.

\bibitem[KS66]{kesten-stigum}
H.~Kesten and B.~P. Stigum.
\newblock A {Limit} {Theorem} for {Multidimensional} {Galton}-{Watson}
  {Processes}.
\newblock {\em The Annals of Mathematical Statistics}, 37(5):1211--1223,
  October 1966.

\bibitem[LLDM08]{leskovec}
J.~Leskovec, K.~J. Lang, A.~Dasgupta, and M.~W. Mahoney.
\newblock Statistical properties of community structure in large social and
  information networks.
\newblock In {\em Proceedings of the 17th International Conference on World
  Wide Web}, WWW '08, pages 695--704. ACM, 2008.

\bibitem[{Mas}14]{massoulie}
L.~{Massouli\'e}.
\newblock {Community detection thresholds and the weak Ramanujan property}.
\newblock {\em Proceedings of the 46th Annual ACM Symposium on Theory of
  Computing (STOC '14)}, pages 694--703, 2014.

\bibitem[{McS}01]{mcsherry}
F.~{McSherry}.
\newblock Spectral partitioning of random graphs.
\newblock In {\em 42nd Annual Symposium on Foundations of Computer Science,
  {FOCS} 2001}, pages 529--537, 2001.

\bibitem[MM05]{mezard-montanari-reconstruction}
M.~M\'ezard and A.~Montanari.
\newblock Reconstruction on {Trees} and {Spin} {Glass} {Transition}.
\newblock {\em Journal of Statistical Physics}, 124(6), 2005.

\bibitem[MMV12]{mmv}
K.~{Makarychev}, Y.~{Makarychev}, and A.~{Vijayaraghavan}.
\newblock Approximation algorithms for semi-random partitioning problems.
\newblock In {\em Proceedings of the 44th Symposium on Theory of Computing
  Conference ({STOC})}, pages 367--384, 2012.

\bibitem[MMV15a]{mmv15}
K.~{Makarychev}, Y.~{Makarychev}, and A.~{Vijayaraghavan}.
\newblock Correlation clustering with noisy partial information.
\newblock In {\em Proceedings of The 28th Conference on Learning Theory
  ({COLT})}, pages 1321--1342, 2015.

\bibitem[MMV15b]{mmv-robust}
K.~Makarychev, Y.~Makarychev, and A.~Vijayaraghavan.
\newblock Learning communities in the presence of errors.
\newblock {\em arXiv:1511.03229}, 2015.

\bibitem[MNS13]{mns2}
E.~{Mossel}, J.~{Neeman}, and A.~{Sly}.
\newblock {A Proof Of The Block Model Threshold Conjecture}.
\newblock {\em arXiv:1311.4115}, November 2013.

\bibitem[MNS14a]{MNSbelief}
E.~{Mossel}, J.~{Neeman}, and A.~{Sly}.
\newblock Belief propagation, robust reconstruction and optimal recovery of
  block models.
\newblock In {\em Proceedings of The 27th Conference on Learning Theory, {COLT}
  2014}, pages 356--370, 2014.

\bibitem[MNS14b]{mns}
E.~{Mossel}, J.~{Neeman}, and A.~{Sly}.
\newblock Reconstruction and estimation in the planted partition model.
\newblock {\em Probability Theory and Related Fields}, pages 1--31, 2014.

\bibitem[Mos98]{mossel-recursive}
E.~Mossel.
\newblock Recursive reconstruction on periodic trees.
\newblock {\em Random Struct. Algorithms}, 1998.

\bibitem[Mos04]{mossel-survey}
E.~Mossel.
\newblock Survey: {Information} flow on trees.
\newblock {\em arXiv:math/0406446}, June 2004.

\bibitem[MS10]{msc}
C.~{Mathieu} and W.~{Schudy}.
\newblock Correlation clustering with noisy input.
\newblock In {\em Proceedings of the Twenty-First Annual {ACM-SIAM} Symposium
  on Discrete Algorithms ({SODA})}, pages 712--728, 2010.

\bibitem[MS15]{ms}
A.~{Montanari} and S.~{Sen}.
\newblock {Semidefinite Programs on Sparse Random Graphs}.
\newblock {\em arXiv:1504.05910}, November 2015.

\bibitem[PW15]{pw}
W.~Perry and A.~S. Wein.
\newblock A semidefinite program for unbalanced multisection in the stochastic
  block model.
\newblock {\em arXiv:1507.05605}, July 2015.

\bibitem[Sly11]{sly-potts}
A.~Sly.
\newblock Reconstruction for the {Potts} model.
\newblock {\em The Annals of Probability}, 39(4):1365--1406, July 2011.

\end{thebibliography}

\appendix

\section{Explicit Computation of the Lower Bound}\label{sec:lower-bound-explicit}

Here we compute an explicit lower bound on the separation $k^6 - k'(k,\eps)^6$ between the average branching factor in the random model and that of the six-level-periodic branching rule described at the end of Section~\ref{sec:treehard}. We will need to assume $k \geq 9$.

Fix $k,\eps,\delta$. We will further modify the six-level-periodic rule by considering a node to be \good if it has three \emph{children} of degree not equal to $2$ and we disregard the contribution of the parent to this rule, to simplify the computation; this only further reduces cutting.

In this one-period tree, the `base' levels are $0$ (the root) and $6$ (the leaves), and the `marking' level is level $3$. In expectation, there are $k^2$ nodes on level $2$. Let us focus on the subtree descending from one of these nodes. Let $p = \prob{\Pois(k) = 1}$; then by Poissonization, before cutting, each node $v$ on level $2$ gives birth to $\Pois(kp)$ degree-$2$ children and $\Pois(k(1-p))$ degree-not-$2$ children. In particular, $v$ is \good with probability $\prob{\Pois(k(1-p)) \geq 3}$. Supposing that this is the case, let us see how many leaves get cut. In expectation, $v$ has $kp$ degree-$2$ children on level $3$, each with exactly one child on level $4$, which is \good with probability $\prob{\Pois(k(1-p)) \geq 3}$. In that case, it gives rise to $kp + \EE[\Pois(k(1-p)) \given\; \geq 3] \, \EE[\Pois(k) \given\; \ne 1]$ leaves, in expectation; but the level-$3$ degree-$2$ node and its subtree gets cut with probability $\delta \eps^2$.

The expected number of leaves cut in the subtree descending from a level-$0$ node is:
\begin{align*}
k^6-(k')^6 &\ge k^2 \cdot \prob{\Pois(k(1-p)) \geq 3}^2 \cdot kp \cdot \delta \eps^2 \cdot (kp + \EE[\Pois(k(1-p)) \given\; \geq 3] \, \EE[\Pois(k) \given\; \ne 1]) \\
& = k^3 p \delta \eps^2 \cdot \prob{\Pois(k(1-p)) \geq 3}^2 \cdot (kp + \EE[\Pois(k(1-p)) \given\; \geq 3] \, \EE[\Pois(k) \given\; \ne 1]).
\end{align*}

\noindent Whenever $k \geq 9$ and $k(1-2 \eps)^2 > 1$, we must have $\eps \geq \frac13$ and $\delta \eps^2 = (1-2\eps)^2 \geq \frac{1}{k}$ (from the definition (\ref{eq:delta}) of $\delta$). In this range we have
$$ k^6 - (k')^6 \geq k^2 p \cdot \prob{\Pois(k(1-p)) \geq 3}^2 \cdot (kp + \EE[\Pois(k(1-p)) \given\; \geq 3] \, \EE[\Pois(k) \given\; \ne 1]) \defeq \mathcal{K}(k). $$
Now, choosing any $\eps$ such that
$$ (k^6 - \mathcal{K}(k))^{1/6} < \frac{1}{(1-2\eps)^2} < k, $$
the semirandom model will be impossible while the random model is possible. This concludes our explicit computation of a separation. We have not made any attempt to optimize it, and that remains an interesting open question for future work.

\section{Proof of Lemma~\ref{lemma:binom-ratio}}\label{sec:binom-ratio-proof}

Recall that $m$ is the number of isolated \marked nodes in $G$, and $g$ is the number of \good nodes. Let $\alpha$ be the number of excess $+1$ spins among the isolated \marked nodes, i.e.\ there are $m/2+\alpha$ \marked isolated nodes of spin $+1$ and $m/2-\alpha$ of spin $-1$. Similarly let $\beta$ be the number of excess $+1$ spins among the \good nodes. Write $\alpha = \alpha_A + \alpha_{BC}$ and $\beta = \beta_A + \beta_{BC}$, where $\alpha_A$ denotes the number of excess $+1$ spins among the isolated \marked nodes of $A$, and likewise for the others. We will need a number of results on the sizes of these values.

\begin{lemma}
\label{lemma:conc}
We have the following results a.a.s.\:
\begin{itemize}
\item $|\alpha_A|,|\beta_A| = O(n^{1/8})$,
\item $|\alpha_{BC}|,|\beta_{BC}| = O(n^{1/2} \log^5 n)$,
\item $m, g = \Theta(n)$.
\end{itemize}
\end{lemma}

\begin{proof}
The first result is easy: $\alpha_A$ and $\beta_A$ are bounded above by $|A|$, which is $O(n^{1/8})$ by assumption. We will establish the remaining concentration results through the bounded differences method; however, this bounded difference property will require controlling the maximum degree of a node. Note that a.a.s.\ every node of $\Gpre$ (and thus $G$) has degree at most $\log n$, by a Chernoff-and-union argument. Let $\trunc(\Gpre)$ denote the graph obtained from $\Gpre$ by simultaneously removing all edges incident to nodes of degree larger than $\log n$. Note now that the radius-$4$ neighborhood of any node in $\trunc(\Gpre)$ has size at most $O(\log^4 n)$.

Let $\mathcal{C}(\sigma,\Gpre)$ denote the census (number of $+1$ spins minus number of $-1$ spins) among those nodes that are `cuttable' in $\trunc(\Gpre)$, i.e.\ nodes $v$ that are degree-$2$ in $\trunc(\Gpre)$, with two opposite-sign (to $v$) neighbors that each have at least $3$ non-degree-$2$ neighbors. This property of $v$ depends only on the radius-$3$ neighborhood in $\trunc(\Gpre)$. Given any vertex $u$ in $\Gpre$, if we add or remove any number of edges incident to $u$ in $\Gpre$, we can only change $\mathcal{C}(\sigma,\Gpre)$ by at most $O(\log^4 n)$, as we can only change the `cuttable' status of the previous and new radius-$4$ neighborhoods of $v$ in $\trunc(\Gpre)$. (Here we need radius $4$ instead of $3$ because by changing edges incident to $v$ we can push the neighbors of $v$ over the $(\log n)$-degree cutoff.) This constitutes a bounded differences property for the function $\mathcal{C}$.

We now apply the following concentration inequality:
\begin{lemma}
Let $F$ be a function of a random graph of size $n$ with independent edges (not necessarily identically distributed). Suppose that, when we add or remove any number of edges incident to any given vertex $v$, the value of $F$ changes by at most $c$. Then
$$ \prob{| F - \EE[F] \given \geq \lambda c \sqrt{n}} \leq 2 \exp(-\lambda^2/2). $$
\end{lemma}
This result is classical, and is based on applying Azuma--Hoeffding to a \emph{vertex exposure martingale}; see for example Chapter~7 of \cite{alon-spencer}.
Note that $\EE[\mathcal{C}] = 0$ due to spin-reversal symmetry. It follows that $|\mathcal{C}| \leq \sqrt{n} \log^5 n$ a.a.s. Moreover, $\trunc(\Gpre) = \Gpre$ a.a.s., so that the true census of cuttable nodes in $\Gpre$ is at most $\sqrt{n} \log^5 n$ a.a.s. Our adversary turns a subset of these cuttable nodes into isolated \marked nodes, choosing each independently with probability $\delta$; applying Hoeffding's inequality separately to the $+1$ and $-1$ cuttable nodes, we see that the census of \marked isolated nodes in $G$ is $O(\sqrt{n} \log^5 n)$ a.a.s. But this census equals $\alpha_A + \alpha_{BC}$, and $|\alpha_A| \leq O(n^{1/8})$, so we have $|\alpha_{BC}| = O(\sqrt{n} \log^5 n)$ a.a.s.
The same argument applies (with a strict subset of the technical mess) to show that $\beta_{BC} = O(\sqrt{n} \log^5 n)$, and that $m$ and $g$ concentrate within $O(\sqrt{n} \log^5 n)$ of their expectations.

It is also straightforward to see that $\EE[m]$ and $\EE[g]$ are $\Theta(n)$; this amounts to showing that the probability of any vertex being \marked isolated in $G$, or \good in $G$, is bounded above zero as $n \to \infty$. But these two properties are local, depending only on a neighborhood in the graph, so this is clear. This completes the proof of the concentration results (Lemma~\ref{lemma:conc}).
\end{proof}

Now we proceed to the proof of Lemma~\ref{lemma:binom-ratio}. We need to show that for a.a.e.\ $(\sigma_{BC},G)$, it holds for all $\sigma_A,\sigma'_A$ that
$$ \frac{ \binom{g/2 + \beta_A + \beta_{BC}}{2}^{m/2 - \alpha_A - \alpha_{BC}} \binom{g/2 - \beta_A - \beta_{BC}}{2}^{m/2 + \alpha_A + \alpha_{BC}}}{ \binom{g/2 + \beta_A' + \beta_{BC}}{2}^{m/2 - \alpha_A' - \alpha_{BC}} \binom{g/2 - \beta_A' - \beta_{BC}}{2}^{m/2 + \alpha_A' + \alpha_{BC}}} = 1 + o(1). $$

We will need sufficiently tight asymptotics for powers of binomials in this regime:
\begin{align*}
& \log \binom{g/2 + \beta_A + \beta_{BC}}{2}^{m/2 - \alpha_A - \alpha_{BC}} \\
&= \left(\frac m2 - \alpha_A - \alpha_{BC}\right) \left(\log\left(\frac g2 + \beta_A + \beta_{BC}\right) + \log\left(\frac g2 + \beta_A + \beta_{BC} - 1\right) - \log 2\right) \\
&= \left(\frac m2 - \alpha_A - \alpha_{BC}\right) \left( 2 \log \frac g2 + \log\left(1 + \frac{\beta_A + \beta_{BC}}{g/2}\right) + \log\left(1 + \frac{\beta_A + \beta_{BC} - 1}{g/2}\right) - \log 2 \right) \\
&= \left(\frac m2 - \alpha_A - \alpha_{BC}\right) \left( 2 \log \frac g2 + \frac{\beta_A + \beta_{BC}}{g/2} + \frac{\beta_A + \beta_{BC} - 1}{g/2} - 4\frac{\beta_{BC}^2}{g^2} - \log 2 \right) + O(n^{-3/8} \log^5 n), \\
\intertext{from the Taylor series for the logarithm,}
&= \left( \frac m2 - \alpha_A - \alpha_{BC} \right) \left( 2 \log \frac g2 - \log 2 \right) + \frac m2 \left( \frac{2\beta_A + 2\beta_{BC} - 1}{g/2} - 4 \frac{\beta_{BC}^2}{g^2}\right) - 2 \frac{\alpha_{BC} \beta_{BC}}{g/2} + O(n^{-3/8} \log^5 n).
\end{align*}
We now apply these asymptotics to the task at hand:
\begin{align*}
& \log \frac{ \binom{g/2 + \beta_A + \beta_{BC}}{2}^{m/2 - \alpha_A - \alpha_{BC}} \binom{g/2 - \beta_A - \beta_{BC}}{2}^{m/2 + \alpha_A + \alpha_{BC}}}{ \binom{g/2 + \beta_A' + \beta_{BC}}{2}^{m/2 - \alpha_A' - \alpha_{BC}} \binom{g/2 - \beta_A' - \beta_{BC}}{2}^{m/2 + \alpha_A' + \alpha_{BC}}} \\
&= \left( \frac m2 - \alpha_A - \alpha_{BC} \right) \left( 2 \log \frac g2 - \log 2 \right) + \frac m2 \left( \frac{2\beta_A + 2\beta_{BC} - 1}{g/2} - 4 \frac{\beta_{BC}^2}{g^2}\right) - 2 \frac{\alpha_{BC} \beta_{BC}}{g/2} \\
&\quad + \left( \frac m2 + \alpha_A + \alpha_{BC} \right) \left( 2 \log \frac g2 - \log 2 \right) + \frac m2 \left( \frac{-2\beta_A - 2\beta_{BC} - 1}{g/2} - 4 \frac{\beta_{BC}^2}{g^2}\right) - 2 \frac{\alpha_{BC} \beta_{BC}}{g/2} \\
&\quad - \left( \frac m2 - \alpha_A' - \alpha_{BC} \right) \left( 2 \log \frac g2 - \log 2 \right) - \frac m2 \left( \frac{2\beta_A' + 2\beta_{BC} - 1}{g/2} - 4 \frac{\beta_{BC}^2}{g^2}\right) + 2 \frac{\alpha_{BC} \beta_{BC}}{g/2} \\
&\quad - \left( \frac m2 + \alpha_A' + \alpha_{BC} \right) \left( 2 \log \frac g2 - \log 2 \right) - \frac m2 \left( \frac{-2\beta_A' - 2\beta_{BC} - 1}{g/2} - 4 \frac{\beta_{BC}^2}{g^2}\right) + 2 \frac{\alpha_{BC} \beta_{BC}}{g/2} \\
&\quad + O(n^{-3/8} \log^5 n) \\
&= O(n^{-3/8} \log^5 n) = o(1),
\end{align*}
as every non-error term cancels. The result now follows: if the logarithm of an expression is $o(1)$ then the expression itself is $1+o(1)$.

\section{Proof of Equation~(\ref{eq:Omega-fact})}\label{sec:Omega-fact-proof}

Enumerate all possible values $E_1, \ldots, E_m$ for $(\sigma_D,G)$. Let $x_i = \prob{E_i}$ and let $y_i = \prob{E_i \text{ and } \sigma\in\Omega}$. We know $\sum_i x_i = 1$ and $\sum_i y_i = \prob{\sigma \in \Omega} = 1 - \epsilon_1$ for some $\epsilon_1 = o(1)$. We want to show that for some $\epsilon_2 = o(1)$ we have $\frac{y_i}{x_i} \ge 1-\epsilon_2$ with probability $p = 1-o(1)$; here the probability is over $i$ drawn proportional to $x_i$. If $I$ is the set of $i$ for which $y_i/x_i \ge 1 - \epsilon_2$, we have $p = \sum_{i \in I} x_i$. This means
$$1-\epsilon_1 = \sum_i y_i = \sum_i x_i \frac{y_i}{x_i} = \sum_{i \in I} x_i \frac{y_i}{x_i} + \sum_{i \notin I} x_i \frac{y_i}{x_i} \le p + (1-p)(1-\epsilon_2)$$
and so $p \ge 1 - \frac{\epsilon_1}{\epsilon_2}$. We need $p = 1-o(1)$ and $\epsilon_2 = o(1)$ so it suffices to take any $\epsilon_2$ such that $\epsilon_2 \to 0$ and $\frac{\epsilon_1}{\epsilon_2} \to 0$, i.e.\ $\epsilon_2$ goes to 0 slower than $\epsilon_1$ does.

\section{Extensions to Dissortative Models}\label{sec:dissortative}

Throughout this paper, we have assumed $a>b$; such block models are often called `assortative'. However, everything does carry through equally in the `dissortative' case $b > a$. Here we briefly sketch the relevant changes.

Our semirandom models are certainly designed for an assortative model, and are entirely too powerful in the dissortative case --- for example, by randomly adding and removing edges appropriately, the current semirandom block model can simulate the Erd\H{o}s--R\'enyi distribution $G(n,k)$ when starting from any dissortative model, which clearly reveals no community structure. Instead, the semirandom model in the dissortative case should add edges between communities, and remove edges within communities, so that these monotone changes are aligned with the latent structure to be recovered. Similarly, the semirandom tree models are able to cut or replace any subtree that follows a same-spin edge.

This leads to many sign changes throughout the paper. We can still couple the resulting graph neighborhoods to a tree distribution; although this tree distribution might look quite different at a glance, it couples perfectly with an assortative tree model, corresponding via $\eps \mapsto 1-\eps$, by flipping all spins at odd levels of the tree. Thus we obtain a random vs.\ semirandom separation for the dissortative tree model. Our recovery results for recursive majority in Section~\ref{sec:recmaj} carry through this coupling also, guaranteeing robust recovery by recursive anti-majority in the dissortative tree model.

Much of the ``no long-range correlations'' argument of Section~\ref{sec:no-lrc} carries through unchanged, but precursors $\Gpre$ of $G$ now reconnect \marked isolated nodes with two \emph{same-spin} \good nodes. Hence the new formula for $|L(\sigma,G)|$ is
$$ |L(\sigma,G)| = \binom{g/2 + \beta}{2}^{m/2 + \alpha} \binom{g/2 - \beta}{2}^{m/2 - \alpha}, $$
which still satisfies Lemma~\ref{lemma:binom-ratio}, by negating each $\beta$ variable everywhere in the proof in Appendix~\ref{sec:binom-ratio-proof}. So we also obtain a random vs.\ semirandom separation in the dissortative block model. 

The SDP upper bounds in Section~\ref{sec:upper-bound} carry through with very few changes of sign. One verifies that the unchanged reference objective $\frac{a-b}{2n} \sigma \sigma^\top$ does maximize $\langle S,- \rangle$ where $S$ is the matrix form of any semirandom change as redefined above. A sign does flip in the proof of Proposition~\ref{prop:new-sdp}: we now have
$$ \gamma_v = \frac{b-a}{2} - O(\log^2 n / \sqrt{n}). $$
Overall we obtain the same guarantees on semirandom partial recovery as in the assortative model, requiring $b > 20$ rather than $a > 20$.

\end{document}